\documentclass[12pt]{amsart}
\usepackage{amsmath, color}
\usepackage{amssymb}
\usepackage{amsfonts}
\usepackage{amsthm}
\usepackage{mathrsfs}
\usepackage[all]{xy}
\usepackage{hyperref}
\hypersetup{
	colorlinks=true,
	citecolor=black,
	filecolor=black,
	linkcolor=black,
	urlcolor=black
}

\newcommand{\be}{\begin{equation}}
\newcommand{\ee}{\end{equation}}
\newtheorem{theorem}{Theorem}[section]
\newtheorem{lemma}[theorem]{Lemma}

\newtheorem{proposition}[theorem]{Proposition}
\newtheorem{corollary}[theorem]{Corollary}

\theoremstyle{definition}
\newtheorem{definition}[theorem]{Definition}

\newtheorem{remark}[theorem]{Remark}

\theoremstyle{conjecture}

\numberwithin{equation}{section}

\makeatletter
\renewcommand{\pod}[1]{\allowbreak\mathchoice
	{\if@display \mkern 18mu\else \mkern 8mu\fi (#1)}
	{\if@display \mkern 18mu\else \mkern 8mu\fi (#1)}
	{\mkern4mu(#1)}
	{\mkern4mu(#1)}
}
\makeatother

\begin{document}

\title{Rota-Baxter operators on Witt and Virasoro algebras}

\author{Xu Gao}
\address{Chern Institute of Mathematics and LPMC, Nankai University, Tianjin 300071, China}
\email{gausyu@gmail.com} 
\author{Ming Liu}
\address{School of Mathematics, South China University of Technology, Guangzhou, Guangdong 510640, China}
\email{ming.l1984@gmail.com}
\author{Chengming Bai}
\address{Chern Institute of Mathematics and  LPMC, Nankai University, Tianjin 300071, China}
\email{baicm@nankai.edu.cn}
\author{Naihuan Jing*}
\address{
Department of Mathematics, North Carolina State University, Raleigh, NC 27695, USA}
\email{jing@math.ncsu.edu}

\thanks{{\scriptsize
\hskip -0.4 true cm MSC (2010): Primary: 17B30; Secondary: 17B68.
\newline Keywords: Rota-Baxter operator, Witt algebra, Virasoro algebra, pre-Lie algebra, PostLie algebra.\\
$*$Corresponding author.
}}

\maketitle
\begin{abstract}
The homogeneous Rota-Baxter operators on the Witt and Virasoro algebras are classified. As applications, the induced solutions of the classical Yang-Baxter equation and the induced pre-Lie and PostLie algebra structures are obtained.
\end{abstract}

\tableofcontents

\section{Introduction}

Rota-Baxter operators were originally defined on associative algebras by G. Baxter to solve an analytic formula in probability~ \cite{Baxter} and then developed by the Rota school \cite{Rota-Kung}. These operators have showed up in many areas in mathematics and mathematical physics such as number theory, combinatorics, operads and quantum field theory (see \cite{Guo1, Guo2} and the references therein).

Rota-Baxter operators in the context of Lie algebras were developed with different motivation. 
In fact, Semenov-Tian-Shansky's fundamental work \cite{S} shows that a Rota-Baxter operator of weight 0 on a Lie algebra is exactly the operator form of the classical Yang-Baxter equation (CYBE), which was regarded as a ``classical limit'' of the quantum Yang-Baxter equation \cite{Belavin}, whereas the latter is also
an important topic in many fields such as symplectic geometry, integrable systems, quantum groups and quantum field theory~(see~\cite{CP} and the references therein).

The study of Rota-Baxter operators on Lie algebras has practical meanings.
First, both Rota-Baxter operators of weight 0 and 1 on a Lie algebra $\frak g$ give
 rise to solutions of CYBE on the double Lie algebra
$\frak g\ltimes_{{\rm ad}^{\ast}}\frak g^{\ast}$ over the direct sum $\frak g\oplus \frak g^{\ast}$ of the Lie algebra $\frak g$ and its dual space $\frak g^{\ast}$. Note
that such a relationship holds for any Lie algebra, which is different from the correspondence given by Semenov-Tian-Shansky with a strict constraint on the Lie algebra itself. Second, there are certain interesting algebraic structures
coming out of the Rota-Baxter operators, notably the pre-Lie algebras from Rota-Baxter operators  of weight 0 on Lie algebras and the PostLie algebras from Rota-Baxter operators of weight 1 on the Lie algebras.  Pre-Lie algebras
are a class of non-associative algebras emerged from the
study of convex homogeneous cones, affine manifolds and deformations of associative algebras \cite{Ko,G,Vi}.
PostLie algebras were introduced in the context of operads \cite{Va}. These two algebraic structures
have appeared in many other fields in mathematics and mathematical physics (see \cite{Bu, BGN2} and the references therein).

Most of the study on Rota-Baxter operators has been focused on the finite dimensional case. For example, a detailed study of Rota-Baxter operators of weight 0 on $\mathfrak{sl}_2(\mathbb C)$ is available in \cite{PBG}. It is natural to consider the infinite dimensional case. As a step in this direction, we study Rota-Baxter operators on
two important infinite dimensional Lie algebras: the Witt algebra and its central extension the Virasoro algebra.
These two Lie algebras are selected due to their important position
in several areas of mathematics and physics.  The following three
problems are addressed in this paper:
\begin{enumerate}
\item  The classification of homogeneous Rota-Baxter operators of weight 0 and 1 on the Witt algebra $W$ and the Virasoro algebra $V$ respectively.
\item The induced solutions of the CYBE on the Lie algebras $W\ltimes_{{\rm ad}^{\ast}} W^{\ast}$ and $V\ltimes_{{\rm ad}^{\ast}} V^{\ast}$ respectively.
\item Description of the induced pre-Lie and PostLie algebra structures respectively.
\end{enumerate}

We note that the study in (2) and (3) can be regarded as applications of the classification results given in (1).

Our results can be briefly summarized as follows. In Section 2, we classify the homogeneous Rota-Baxter operators
of weight 0 and 1 on the Witt algebra $W$. In Section 3, we classify the homogeneous Rota-Baxter
operators of weight 0 and 1 on the Virasoro algebra $V$.
In particular, we find that although homogeneous Rota-Baxter operators with degree $0$
	on the Witt algebra and Virasoro algebra are closely related,
	those operators with nonzero degree are
	quite different in the sense that
	the former operators are not special cases of the latter, which
	is primarily a result of the central element $C$ in the Virasoro algebra.
In Section 4, we give the induced solutions of the CYBE on the Lie algebras $W\ltimes_{{\rm ad}^{\ast}} W^{\ast}$ and $V\ltimes_{{\rm ad}^{\ast}} V^{\ast}$ respectively.
In Section 5, we give the induced pre-Lie algebras from the Rota-Baxter operators of weight 0 on the Witt algebra
$W$ and Virasoro algebra $V$ respectively.  In Section 6, we give the induced PostLie algebras from the Rota-Baxter operators of weight 1 on $W$ and $V$ respectively.

\section{Homogeneous Rota-Baxter operators on the Witt algebra}

\begin{definition} Let $\mathbb F$ be a field. A {\it Rota-Baxter operator of weight $\lambda \in \mathbb F$} on a Lie algebra $\mathfrak{g}$
over $\mathbb F$ is 
a linear map $R: \mathfrak{g}\rightarrow \mathfrak{g}$ satisfying
 \begin{equation}\label{RB-lambda}
   [R(x),R(y)] = R([R(x),y]+[x,R(y)])+\lambda R([x,y]),\;\;\forall x, y\in \mathfrak{g}.
 \end{equation}
\end{definition}
Note that if $R$ is a Rota-Baxter operator of weight $\lambda \neq 0$, then ${\lambda}^{-1}R$ is a Rota-Baxter
operator $R$ of weight $1$. Therefore one only needs to consider
Rota-Baxter operators of weight $0$ and $1$. 
We also assume that $\mathbb F=\mathbb C$, the complex field since
both the Witt and Virasoro algebras are defined over $\mathbb C$.

\begin{definition}
The {\it Witt algebra} $W$ is a Lie algebra with the basis $\{L_n| n\in \mathbb{Z}\}$ subject to the following relations:
  \begin{equation}\label{Witt}
    [L_m,L_n]=(m-n)L_{m+n}, \;\; \forall m,n \in \mathbb{Z}.
  \end{equation}
\end{definition}

There is a natural $\mathbb{Z}$-grading on the Witt algebra $W$, namely
\[
W=\bigoplus_{n\in\mathbb{Z}}W_n,
\]
where $W_n=\mathbb{C}L_n$ for any $n\in \mathbb{Z}$.

\begin{definition}\label{def:homoWitt}
Let $k$ be an integer. A {\it homogeneous operator $F$ with degree $k$} on the Witt algebra $W$ is a linear operator on $W$ satisfying
\[
F(W_m)\subset W_{m+k},\;\;\forall m\in\mathbb{Z}.
\]
\end{definition}

Therefore, a {\it homogeneous Rota-Baxter operator $R_k$ with degree $k$} on the Witt algebra $W$ is a Rota-Baxter operator on $W$
of the following form
\begin{equation}\label{eq:HRB}
R_k(L_m)=f(m+k)L_{m+k}, \;\;\forall m\in \mathbb{Z},
\end{equation}
where $f$ is a $\mathbb C$-valued function defined on $\mathbb{Z}$.

\subsection{Homogeneous Rota-Baxter operators of weight $0$ on the Witt algebra}
\mbox{}
Let $R_k$ be a homogeneous Rota-Baxter operator of weight $0$ with degree $k$ on the Witt algebra
$W$ satisfying Eq.~(\ref{eq:HRB}). Then by Eqs.~(\ref{RB-lambda}) and (\ref{Witt}), we see that the function $f$ satisfies the following equation:
\begin{equation}\label{weight 0}
  f(m)f(n)(m-n)=f(m+n)(f(m)(m-n+k)+f(n)(m-n-k)),\;\;\forall m,n\in \mathbb{Z}.
\end{equation}

\begin{proposition}\label{prop1}
With the notations as above, the Rota-Baxter operator $R_0$ of weight $0$ with degree $0$
is given by 
\begin{equation*}
f(m)=\alpha\delta_{m,0},\;\; \forall m\in\mathbb{Z},
\end{equation*}
for $\alpha\in \mathbb{C}$.
\end{proposition}

\begin{proof}
When $k=0$, Eq.~(\ref{weight 0}) becomes
\begin{equation*}
(m-n)f(m)f(n)=(m-n)f(m+n)(f(m)+f(n)),\;\;\forall m,n\in \mathbb{Z}.
\end{equation*}
Plugging $n=0$ in the equation, we have
\begin{equation*}
mf(m)^2=0.
\end{equation*}
Thus $f(m)=\alpha\delta_{m,0}$ for some $\alpha\in \mathbb{C}$.
\end{proof}

When $k\neq 0$, taking $n=0$ in Eq.~(\ref{weight 0}), we have

\begin{equation}\label{Eq2.6}
  f(m)((m+k)f(m)-kf(0))=0,\;\;\forall m\in \mathbb{Z}.
\end{equation}

\begin{proposition}\label{prop2}
With the notations as above, when the degree $k\in\mathbb{Z}^{\ast}:=\mathbb{Z}\setminus\{0\}$ and $f(0)=0$, we have
\begin{equation*}
f(m)=\alpha\delta_{m+k,0},\;\;\forall m\in \mathbb{Z},
\end{equation*}
where $\alpha\in \mathbb{C}^{\ast} := \mathbb{C}\setminus\{0\}$.
\end{proposition}

\begin{proof}
If $f(0)=0$, then  by Eq. (\ref{Eq2.6}),
we have
\begin{equation*}
  (m+k)(f(m))^2=0,\;\;\forall m\in \mathbb{Z}.
\end{equation*}
Thus, the function $f$ satisfies
\begin{align*}
f(m)=\alpha\delta_{m+k,0},\;\;\forall m\in \mathbb{Z},
\end{align*}
where $\alpha\in \mathbb{C}^{\ast}$.
\end{proof}

 When $f(0)\neq 0$, if follows from Eq. (\ref{Eq2.6}) that $f(-k)=0$. Moreover, substituting this into Eq.~(\ref{weight 0}) with $m=k$ and
$n=-k$, we have $f(k)=0$. For such an $f$ satisfying Eq.~(\ref{weight 0}) so that $kf(0)\neq 0$, we set
\[
\mathcal{I}=\{m\in \mathbb{Z}| f(m)=0\},\;\;
\mathcal{J}=\{m\in \mathbb{Z}|(m+k)f(m)-kf(0)=0\}.
\]
Thus $-k, k\in \mathcal{I}$ and $\mathcal{I}\cap \mathcal{J}=\varnothing$, $\mathcal{I}\cup \mathcal{J}=\mathbb{Z}$.

\begin{lemma}\label{lemma2.5}
Let $f$ be a $\mathbb C$-valued function defined on $\mathbb{Z}$ satisfying Eq.~(\ref{weight 0}). Suppose that $f(0)\neq 0$ and $k\ne 0$. If $n\in \mathcal{J}$ and
$m\neq n, n+k$, then $m\in \mathcal{I}$ if and only if $m+n\in \mathcal{I}$.
\end{lemma}

\begin{proof}
If $m\in \mathcal{I}$, $m\neq n+k$ and $n\in \mathcal{J}$, then by Eq.~(\ref{weight 0}), we have
\begin{align*}
(m-n-k)f(n)f(n+m)=0.
\end{align*}
Since $n\in \mathcal{J}$, we have $f(n+m)=0$. Conversely, if $m+n\in \mathcal{I}$, $m\neq n$ and $n\in \mathcal{J}$, then
by Eq.~(\ref{weight 0}), we have
$$(m-n)f(m)=0.$$
Hence $m\in \mathcal{I}$.
\end{proof}

For an integer $m\in \mathbb{Z}$, set
\[
\mathcal{J}_m=\{n\in \mathcal{J}| mn\in \mathcal{J}\},\;\;
\mathcal{I}_m=\{n\in \mathcal{J}| mn+k\in \mathcal{I}\}.
\]

\begin{proposition}\label{prop2.6} With the conditions as above, we have
\begin{enumerate}
  \item $\mathcal{J}_0=\mathcal{J}_1=\mathcal{J}$.
  \item $(\mathcal{J}\setminus \{-\tfrac{k}{2m}\})\cap \mathcal{J}_m\subset \mathcal{J}_{-m}$ for every $m\neq 0$.
  In particular, $\mathcal{J}\setminus \{-\tfrac{k}{2}\}\subset \mathcal{J}_{-1}$.
  \item $(\mathcal{J}\setminus \{-\tfrac{k}{2}, \tfrac{k}{m+1}\})\cap \mathcal{J}_{m-1}\subset \mathcal{J}_m$,
  $(\mathcal{J}\setminus \{\tfrac{-k}{m+1}\})\cap \mathcal{J}_{1-m}\subset \mathcal{J}_{-m}$ for $m\geq 2$.
  \item $(\mathcal{J}\setminus \{\tfrac{k}{2m-1}\})\cap \mathcal{J}_{1-m}\subset \mathcal{J}_{m}$,
  $(\mathcal{J}\setminus \{-\tfrac{k}{2}, \tfrac{-k}{2m-1}\})\cap \mathcal{J}_{m-1}\subset \mathcal{J}_{-m}$ for $m\geq 2$.
\end{enumerate}
\end{proposition}

\begin{proof}
(1) follows immediately by definition.  We only give a proof for (2), as (3) and (4) can be proved similarly.

In fact, it is straightforward to check that $0\in (\mathcal{J}\setminus \{-\tfrac{k}{2m}\})\cap \mathcal{J}_m$
and $0\in \mathcal{J}_{-m}$ for $m\neq 0$.
Let $n\neq 0$ and $n\in (\mathcal{J}\setminus \{-\tfrac{k}{2m}\})\cap \mathcal{J}_m$.
To prove (2), we only need to show
that $-nm\in \mathcal{J}$. Otherwise,  $-nm\in \mathcal{I}$. Then by Lemma~\ref{lemma2.5},  we have  $nm-nm=0\in \mathcal{I}$,
which is a contradiction with the assumption that $f(0)\neq 0$.
\end{proof}
\begin{corollary}\label{cor2.7} With the conditions as above, we have
\begin{align*}
\mathcal{J}\setminus \{\tfrac{-k}{2}\}\subset \bigcap_{m\in \mathbb{Z}} \mathcal{J}_m.
\end{align*}
\end{corollary}

\begin{proof}
We only need to show that $\mathcal{J}\setminus \{\tfrac{-k}{2}\}\subset \mathcal{J}_m\cap \mathcal{J}_{-m}$ for every $m\geq 1$.

By Proposition~\ref{prop2.6}, we show that $\mathcal{J}\setminus \{\tfrac{-k}{2}\}\subset \mathcal{J}_{-1}$. Moreover, since
$\mathcal{J}_1=\mathcal{J}_0=\mathcal{J}$, we have $\mathcal{J}\setminus \{\tfrac{-k}{2}\}\subset \mathcal{J}_1\cap \mathcal{J}_{-1}$.

By Proposition~\ref{prop2.6} again, we have
$$
\mathcal{J}\setminus \{\tfrac{-k}{2},\tfrac{k}{3}\}\subset \mathcal{J}_2,\;\;
\mathcal{J}\setminus \{\tfrac{-k}{2},\tfrac{-k}{3}\}\subset \mathcal{J}_{-2},\;\;
\mathcal{J}\setminus \{\tfrac{k}{4}\}\cap \mathcal{J}_{-2}\subset \mathcal{J}_2,\;\;
\mathcal{J}\setminus \{\tfrac{-k}{4}\}\cap \mathcal{J}_2\subset \mathcal{J}_{-2}.
$$
Therefore
$$\mathcal{J}\setminus \{\tfrac{-k}{2},\tfrac{-k}{3},\tfrac{k}{4}\}=
(\mathcal{J}\setminus \{\tfrac{k}{4}\})\cap (\mathcal{J}\setminus \{\tfrac{-k}{2},\tfrac{-k}{3}\})\subset (\mathcal{J}\setminus \{\tfrac{k}{4}\})\cap \mathcal{J}_{-2}\subset \mathcal{J}_2.$$
Hence $$\mathcal{J}\setminus \{\tfrac{-k}{2}\}=(\mathcal{J}\setminus \{\tfrac{-k}{2},\tfrac{-k}{3},\tfrac{k}{4}\})\cup (\mathcal{J}\setminus \{\tfrac{-k}{2},\tfrac{k}{3}\})\subset \mathcal{J}_2.$$

Similarly, we show that $\mathcal{J}\setminus \{\tfrac{-k}{2}\}\subset \mathcal{J}_{-2}$.

Now assume that $\mathcal{J}\setminus \{\tfrac{-k}{2}\}\subset \mathcal{J}_{m-1}\cap \mathcal{J}_{1-m}$ holds for $m>2$. By Proposition~\ref{prop2.6} we have that
\begin{align*}
\mathcal{J}\setminus \{\tfrac{-k}{2},\tfrac{k}{m+1}\}=(\mathcal{J}\setminus \{\tfrac{-k}{2},\tfrac{k}{m+1}\})\cap (\mathcal{J}\setminus \{\tfrac{-k}{2}\})\subset
(\mathcal{J}\setminus \{\tfrac{-k}{2},\tfrac{k}{m+1}\})\cap \mathcal{J}_{m-1}\subset \mathcal{J}_m,
\end{align*}
and
\begin{align*}
\mathcal{J}\setminus \{\tfrac{-k}{2},\tfrac{k}{2m-1}\}=(\mathcal{J}\setminus \{\tfrac{k}{2m-1}\})\cap (\mathcal{J}\setminus \{\tfrac{-k}{2}\})\subset
(\mathcal{J}\setminus \{\tfrac{k}{2m-1}\})\cap \mathcal{J}_{1-m}\subset \mathcal{J}_m.
\end{align*}
Since $\tfrac{k}{m+1}\neq \tfrac{k}{2m-1}$ for $m>2$, we have
\begin{align*}
\mathcal{J}\setminus \{\tfrac{-k}{2}\}=(\mathcal{J}\setminus \{\tfrac{-k}{2},\tfrac{k}{m+1}\})\cup
(\mathcal{J}\setminus \{\tfrac{-k}{2},\tfrac{k}{2m-1}\})\subset \mathcal{J}_m.
\end{align*}
Similarly we show that $\mathcal{J}\setminus \{\tfrac{-k}{2}\}\subset \mathcal{J}_{-m}$ for $m>2$.
\end{proof}

\begin{proposition}\label{prop2.8} With the conditions as above, we have
\begin{enumerate}
  \item $\mathcal{I}_0=\mathcal{J}$.
  \item $\mathcal{J}\setminus \{\tfrac{-k}{2},\tfrac{k}{m}\}\subset \mathcal{I}_m$ for any $m\neq 0$.
  \item $\mathcal{J}\setminus \{\tfrac{-k}{2},\tfrac{-k}{2m}\}\subset \mathcal{I}_m$ for any $m\neq 0$.
\end{enumerate}
\end{proposition}
\begin{proof}
(1) follows from the fact that $k\in \mathcal{I}$.  We only give a detailed proof of (3) and the proof of (2) is similar.

Let $m$ be a fixed non-zero integer.  Since $0\in \mathcal{J}$ and $k\in \mathcal{I}$, we show that $0\in \mathcal{I}_m$.
Let $n_0$ be an arbitrary nonzero integer in $\mathcal{J}\setminus \{\tfrac{-k}{2},\tfrac{-k}{2m}\}$.
Then we have $k+mn_0\neq -mn_0$ and
$k+mn_0\neq -mn_0+k$. By Corollary~\ref{cor2.7}, we have $-mn_0\in \mathcal{J}$. Hence by Lemma~\ref{lemma2.5} and since $mn_0+k-mn_0=k\in \mathcal{I}$, we have
$mn_0+k\in \mathcal{I}$.
\end{proof}

By Proposition~\ref{prop2.8}, we get the following result.
\begin{corollary}\label{cor2.9} With the conditions as above, we have
\begin{align*}
\mathcal{J}\setminus \{\tfrac{-k}{2}\}\subset \bigcap_{m\in \mathbb{Z}} \mathcal{I}_m.
\end{align*}
\end{corollary}

\begin{proposition}\label{prop2.10} With the conditions as above, let $n\in \mathcal{J}\setminus \{\tfrac{-k}{2}\}$ and $n\ne 0$. Then we have $n\nmid k$, and
for any $m\in \mathbb{Z}$,
\begin{enumerate}
  \item if $m\in \mathcal{I}$, then $m+n\mathbb{Z}\in \mathcal{I}$;
  \item if $m\in \mathcal{J}$, then $m+n\mathbb{Z}\in \mathcal{J}$.
\end{enumerate}
\end{proposition}

\begin{proof}
If $m$ is neither in $n\mathbb{Z}$ nor in $k+n\mathbb{Z}$, the conclusion holds due to Lemma~\ref{lemma2.5}.
On the other hand, by Corollary \ref{cor2.7}, we show that $n\in \bigcap\limits_{l\in \mathbb{Z}} \mathcal{J}_l$.
Hence $n\mathbb{Z}\subset \mathcal{J}$.
Furthermore, by Corollary \ref{cor2.9} and the fact that $k\in \mathcal{I}$, we have $n\in \bigcap\limits_{l\in \mathbb{Z}} \mathcal{I}_l$.
Thus $k+n\mathbb{Z}\subset \mathcal{I}$. Therefore for any $m\in n\mathbb{Z}$ or $m\in k+n\mathbb{Z}$, if $m\in \mathcal{I}$, then $m\in k+n\mathbb{Z}$ and hence
$m+n\mathbb{Z}\in \mathcal{I}$, and if $m\in \mathcal{J}$, then $m\in n\mathbb{Z}$ and hence $m+n\mathbb{Z}\in \mathcal{J}$.
Moreover $n\nmid k$. Otherwise we have $n\mathbb{Z}=k+n\mathbb{Z}\subset \mathcal{I}\cap \mathcal{J}$, which is a contradiction.
\end{proof}

For any $m,n\in\mathbb{Z}$ (not both zero), let ${\rm gcd}(m,n)$ denote the greatest common divisor of $m$ and $n$.

\begin{corollary}\label{cor2.11} With the conditions as above, if $n_1\in \mathcal{J}$, $n_2\in \mathcal{J}\setminus \{0,\tfrac{-k}{2}\}$, then ${\rm gcd}(n_1,n_2)\mathbb{Z}\subset \mathcal{J}$.
\end{corollary}
\begin{proof}
If $n_1\neq \tfrac{-k}{2}$, then by Proposition~\ref{prop2.10}, we show that for every $m_1$, $m_2\in\mathbb{Z}$,
$n_1m_1+n_2m_2\in \mathcal{J}$. Furthermore, we have ${\rm gcd}(n_1,n_2)\mathbb{Z}=n_1\mathbb{Z}+n_2\mathbb{Z}$. Thus
${\rm gcd}(n_1,n_2)\mathbb{Z}\subset \mathcal{J}$.

If $n_1=\tfrac{-k}{2}\in \mathcal{J}$, then by Proposition~\ref{prop2.10}, we show that $n_1+n_2\in \mathcal{J}$.
On the other hand, we have $n_2$, $n_1+n_2\in \mathcal{J}\setminus \{\tfrac{-k}{2}\}$. Hence gcd$(n_1+n_2, n_2)\mathbb{Z}\subset \mathcal{J}$.
Since gcd$(n_1+n_2, n_2)\mathbb{Z}={\rm gcd} (n_1,n_2)\mathbb{Z}$, we have ${\rm gcd} (n_1,n_2)\mathbb{Z}\subset \mathcal{J}$.
\end{proof}

\begin{proposition}\label{prop3}
Let $f$ be a $\mathbb C$-valued function defined on $\mathbb{Z}$ satisfying Eq.~(\ref{weight 0}). Suppose that $f(0)\neq 0$ and $k\ne 0$.
If $\tfrac{-k}{2}\in \mathbb{Z}$ and $\tfrac{-k}{2}\in \mathcal{J}$, then $\mathcal{J}=\{0,\tfrac{-k}{2}\}$, and in this case,
\begin{equation}f(m)=\delta_{m,0}f(0)+2
\delta_{m,\frac{-k}{2}}f(0),\;\;\forall m\in \mathbb{Z}.\end{equation}
\end{proposition}
\begin{proof}
It is obvious that $\{0, \tfrac{-k}{2}\}\subset \mathcal{J}$. Conversely, if there exists an $n_0\in \mathcal{J}$
such that $n_0\neq 0$, $\tfrac{-k}{2}$,  then by Corollary \ref{cor2.11}, we have
${\rm gcd} (n_0,\tfrac{-k}{2})\mathbb{Z}\subset \mathcal{J}$. Since $\mathcal{J}\neq \mathbb{Z}$, we have
${\rm gcd} (n_0,\tfrac{-k}{2})\neq 1$. Set $d={\rm gcd}(n_0,\tfrac{-k}{2})$. Then $d|\tfrac{-k}{2}$. Hence
$d|k$. By Proposition~\ref{prop2.10}, we show that $d=\tfrac{-k}{2}$. Thus $n_0=\tfrac{k}{2}m_0$
for some $m_0\neq 0$, $-1$. However, by Lemma~\ref{lemma2.5} and induction on $m$ (note that $\pm k\in \mathcal{I}$), one can show that
$\tfrac{k}{2}m\in \mathcal{I}$ for any $m\neq0$, $-1$. It is a contradiction.
Hence $\mathcal{J}=\{0,-\tfrac{k}{2}\}$.
\end{proof}

\begin{proposition}\label{prop4}
Let $f$ be a $\mathbb C$-valued function on $\mathbb{Z}$ satisfying Eq.~(\ref{weight 0}). Suppose that $f(0)\neq 0$ and $k\ne 0$.
If $\tfrac{-k}{2}\notin \mathcal{J}$, and $\{0\}\subsetneqq \mathcal{J}$, then
there exists a non-zero integer $l\in \mathcal{J}$, $l\nmid k$ such that $|l|$ is minimal.
In this case, we have
\begin{align*}
\mathcal{J}=l\mathbb{Z},
\end{align*}
and thus
 \begin{equation}
        f(m)=
        \tfrac{k}{m+k}\delta_{m,l\mathbb{Z}}f(0),
    \end{equation}
where
\[
\delta_{m,l\mathbb{Z}}:=
\sum_{n\in\mathbb{Z}}\delta_{m,ln} =
\begin{cases}
1 & m\in l\mathbb{Z};\\
0 & m\notin l\mathbb{Z}.
\end{cases}
\]
\end{proposition}
\begin{proof}
Since $ \{0\}\subsetneqq \mathcal{J}$, there exists an integer $l\in \mathcal{J}$ such that $l\neq 0$, $|l|$ is minimal and $l\nmid k$.
By Proposition~\ref{prop2.10} and the minimality of $|l|$, we have $m\in \mathcal{I}$ for any $m\notin l\mathbb{Z}$.
On the other hand, since $0\in \mathcal{J}$ and by Proposition~\ref{prop2.10} again, we have $l\mathbb{Z}\subset \mathcal{J}$. Hence $\mathcal{J}=l\mathbb{Z}$ and thus the conclusion holds.
\end{proof}

In summary, we obtain the following classification.

\begin{theorem}\label{Thm2.14}
A homogeneous Rota-Baxter operator of weight $0$ on the Witt algebra $W$ must be one of the following operators.
\begin{enumerate}
	\item[(I)] $R_k^{\alpha}(L_m) = \alpha\delta_{m+2k,0}L_{m+k},\;\;\forall m\in\mathbb{Z}$,
	where $k\in\mathbb{Z}$ and $\alpha \in \mathbb C$.
	
	\item[(II)] $R_{2k}^{'\beta}(L_m) = (\beta\delta_{m+2k,0}+2\beta\delta_{m+3k,0})L_{m+2k},\;\;\forall m\in\mathbb{Z}$,
	where $k\in\mathbb{Z}^{\ast}$ and $\beta \in \mathbb{C}^{\ast}$.

	\item[(III)] $R_k^{l,\gamma}(L_m) = \tfrac{k}{m+2k}\gamma\delta_{m+k,l\mathbb{Z}} L_{m+k},\;\;\forall m\in\mathbb{Z}$,
	where $k,l\in\mathbb{Z}^{\ast}$, $l\nmid k$ and $\gamma \in \mathbb{C}^{\ast}$.
\end{enumerate}
Moreover,
\begin{enumerate}
	\item $\left\{R_0^{\alpha}
	\middle|\alpha\in\mathbb{C}
	\right\}$ are all the homogeneous Rota-Baxter
	operators of weight $0$ with degree $0$ on the Witt algebra $W$.
	\item If $k\neq 0$ and is odd, then
	$\left\{R_k^{\alpha}, R_k^{l,\beta}
	\middle|\alpha\in\mathbb{C}, \beta\in\mathbb{C}^{\ast}, l\in\mathbb{Z}^{\ast}, l\nmid k
	\right\}$
	are all the homogeneous Rota-Baxter
	operators of weight $0$ with degree $k$ on $W$.
	\item If $k\neq 0$ and is even, then
	$\left\{R_k^{\alpha}, R_k^{'\beta}, R_k^{l,\gamma}
	\middle|\alpha\in\mathbb{C}, \beta,\gamma\in\mathbb{C}^{\ast}, l\in\mathbb{Z}^{\ast}, l\nmid k
	\right\}$
	are all the homogeneous Rota-Baxter
	operators of weight $0$ with degree $k$ on $W$.
\end{enumerate}
\end{theorem}

\begin{proof}
The first part follows from Propositions~\ref{prop1},
\ref{prop2}, \ref{prop3} and \ref{prop4}. The second part
can be directly verified. 
\end{proof}

\begin{remark}\label{notation1}
	Note that
	\[
	R_k^{\alpha} = \alpha R_k^1, \;\;
	R_{2k}^{'\beta} = \beta R_{2k}^{'1}, \;\;
	R_k^{l,\gamma} = \gamma R_k^{l,1},
	\]
	for any $\alpha\in\mathbb{C}$ and $\beta,\gamma\in\mathbb{C}^{\ast}$.
	This partly explains our notation in the theorem.
\end{remark}

\begin{remark}\label{rek:RB0}
It is known that $R$ is a Rota-Baxter operator of weight $0$ on a Lie algebra $\frak g$ if and only if $\alpha R$ is a Rota-Baxter operator of weight $0$ on $\frak g$ for $0\ne \alpha\in \mathbb C$. So the set of Rota-Baxter operators of weight 0 on any Lie algebra carries an action of $\mathbb C^{\ast}$ by scalar multiplication. In this sense, the above theorem can be rewritten as follows.

A complete set of representatives of all
homogeneous Rota-Baxter operators of weight $0$ with degree $k$ on the Witt algebra $W$ under the action of $\mathbb{C}^{\ast}$ by scalar multiplication is
\begin{itemize}
	\item $\mathcal{R}^W_0=
	\left\{R_0^0=0, R_0^1\right\}$,
	if $k=0$;
	\item $\mathcal{R}^W_k=
	\left\{R_k^0=0, R_k^1, R_k^{l,1}\middle|l\in\mathbb{Z}^{\ast}, l\nmid k\right\}$,
	if $k\neq0$ and is odd;
	\item $\mathcal{R}^W_k=
	\left\{R_k^0=0, R_k^1, R_k^{'1}, R_k^{l,1}\middle|l\in\mathbb{Z}^{\ast}, l\nmid k\right\}$,
	if $k\neq0$ and is even.
\end{itemize}

\end{remark}

\subsection{Homogeneous Rota-Baxter operators of weight $1$ on the Witt algebra}

It is straightforward to show by definition that there does not exist any homogeneous Rota-Baxter operator  of weight $1$ with a nonzero
degree $k$ on the Witt algebra $W$.

Let $R_0$ be a homogeneous Rota-Baxter operator of weight $1$ with degree $0$ on the Witt algebra
$W$ satisfying Eq.~(\ref{eq:HRB}), that is,
\begin{equation}
R_0(L_m)=f(m)L_{m},\;\forall m\in \mathbb{Z}.
\end{equation}
Then by Eqs.~(\ref{RB-lambda}) and (\ref{Witt}), we show that the function $f$ satisfies the following equation:
\begin{equation}\label{Eq2.7}
    f(m)f(n)(m-n)=f(m+n)(f(m)+f(n)+1)(m-n),\;\forall m,n\in\mathbb{Z}.
  \end{equation}

Let $n=0$ in Eq. (\ref{Eq2.7}), then we have
\begin{align*}
 mf(m)(f(m)+1)=0,\;\;\forall m\in \mathbb{Z}.
\end{align*}
Set
\[
\mathcal{I}_1=\{m\in \mathbb{Z}| f(m)=0\},\;\;\mathcal{I}_2=\{m\in \mathbb{Z}|f(m)=-1\}.
\]
\begin{lemma}\label{lemma2.15}
Let $f$ be a $\mathbb C$-valued function defined on $\mathbb{Z}$ satisfying Eq.~(\ref{Eq2.7}). If $m,n\in \mathbb{Z}$ such that $m\ne n$ and $m, n\in \mathcal{I}_i$, then $m+n\in \mathcal{I}_i$ ($i=1,2$).
\end{lemma}
\begin{proof}
If $m\neq n$ and $m,n \in \mathcal{I}_1$, then Eq.~(\ref{Eq2.7}) implies $m+n\in \mathcal{I}_1$.
Similarly, if $m\neq n$ and $m,n \in \mathcal{I}_2$, then $m+n\in \mathcal{I}_2$.
\end{proof}

\begin{proposition}\label{prop2.16}
With the notations as above, if there exists a nonzero integer $m_0$ such that $m_0,-m_0\in \mathcal{I}_1$, then $\mathcal{I}_1,\mathcal{I}_2$  must be one of the following cases:
  \begin{enumerate}
    \item $\mathcal{I}_1=\{m\mid m\leqslant1\}, \mathcal{I}_2=\{m\mid m\geqslant2\}$;
    \item $\mathcal{I}_1=\{m\mid m\geqslant-1\}, \mathcal{I}_2=\{m\mid m\leqslant-2\}$;
    \item $\mathcal{I}_1=\mathbb{Z}, \mathcal{I}_2=\varnothing$.
  \end{enumerate}
\end{proposition}
\begin{proof}
By Lemma~\ref{lemma2.15}, $0=m_0+(-m_0)\in \mathcal{I}_1$.  Hence
Eq.~(\ref{Eq2.7}) with $n=-m\ne 0$ implies
$$f(m)f(-m)=0,\;\;\forall m\neq 0.$$
Therefore if $m\ne 0$ and $m \in \mathcal{I}_2$, then $-m\in \mathcal{I}_1$.

Let $l$ be the minimal positive integer such that $l\in \mathcal{I}_1$. Then $l=1$. Otherwise, $l\geqslant2$. So $1\in \mathcal{I}_2$. Hence $-1\in \mathcal{I}_1$. By Lemma~\ref{lemma2.15}, we show that $l-1=-1+l\in \mathcal{I}_1$ which contradicts with the minimality of $l$. Similarly, $-1\in \mathcal{I}_1$.

 \begin{enumerate}
    \item If $2\in \mathcal{I}_2$, then $-2\in \mathcal{I}_1$. Since $-1,-2\in \mathcal{I}_1$, by Lemma~\ref{lemma2.15} and induction, $\mathcal{I}_1$ contains all negative integers.
   Thus for any $m\geqslant2$, $2-m\in \mathcal{I}_1$. It implies $m\in \mathcal{I}_2$. Otherwise, by Lemma~\ref{lemma2.15}, $2=(2-m)+m\in \mathcal{I}_1$ which contradicts with the assumption that $2\in \mathcal{I}_2$. In this case, $\mathcal{I}_1=\{m\mid m\leqslant1\}, \mathcal{I}_2=\{m\mid m\geqslant2\}$.
    \item Similarly, if $-2\in \mathcal{I}_2$, then $\mathcal{I}_1=\{m\mid m\geqslant-1\}, \mathcal{I}_2=\{m\mid m\leqslant-2\}$.
    \item If $2,-2\in \mathcal{I}_1$, then $\mathcal{I}_1=\mathbb{Z}, \mathcal{I}_2=\varnothing$.
  \end{enumerate}
Hence the conclusion holds.
\end{proof}

Similarly, we have the following conclusion.

\begin{proposition}\label{prop2.17}
With the notations as above, if there exists a nonzero integer $m_0$ such that $m_0,-m_0\in \mathcal{I}_2$, then $\mathcal{I}_1,\mathcal{I}_2$ are
one of the following cases:
  \begin{enumerate}
    \item $\mathcal{I}_1=\{m\mid m\geqslant 2\}, \mathcal{I}_2=\{m\mid m\leqslant 1\}$;
    \item $\mathcal{I}_1=\{m\mid m\leqslant -2\}, \mathcal{I}_2=\{m\mid m\geqslant-1\}$;
    \item $\mathcal{I}_1=\varnothing, \mathcal{I}_2=\mathbb{Z}$.
  \end{enumerate}
\end{proposition}

\begin{proposition}\label{prop2.18}
With the notations as above, if there does not exist a nonzero integer $m$ such that $m,-m\in \mathcal{I}_i$, $i=1,2$, then either
    \begin{equation*}
      \mathbb{Z}_+\subset \mathcal{I}_1, \mathbb{Z}_-\subset \mathcal{I}_2,
    \end{equation*}
    or
    \begin{equation*}
      \mathbb{Z}_-\subset \mathcal{I}_1, \mathbb{Z}_+\subset \mathcal{I}_2,
    \end{equation*}
where $\mathbb{Z}_+$ (resp. $\mathbb{Z}_-$) denotes the set of positive (resp. negative) integers.
Moreover, $f(0)\in \mathbb C$ is arbitrary.
\end{proposition}
\begin{proof}
 In this case, $f(m)\neq f(-m)$ for every $m\neq0$. Thus $m\in \mathcal{I}_1$ if and only if $-m\in \mathcal{I}_2$. So the conclusion about $\mathcal{I}_1, \mathcal{I}_2$ holds.
Moreover, Eq.~(\ref{Eq2.7}) holds automatically when we set $n=0$ or $m+n=0$, that is, $f(0)\in \mathbb C$ is arbitrary.
\end{proof}

Summarizing Propositions~\ref{prop2.16}, \ref{prop2.17} and
\ref{prop2.18} and with a similar discussion as that of
Theorem~\ref{Thm2.14}, we obtain the following classification.

\begin{theorem}\label{Thm2.21}
	A homogeneous Rota-Baxter operator of weight $1$ with degree $0$ on the Witt algebra $W$ must be one of the following types.
	\begin{enumerate}
		\item
		$
		R_0^{\leqslant1}(L_m)=
		\begin{cases}
		-L_m & m\geqslant2;\\
		0 & m\leqslant1.
		\end{cases}
		$
		\item
		$
		R_0^{\geqslant-1}(L_m)=
		\begin{cases}
		-L_m & m\leqslant-2;\\
		0 & m\geqslant-1.
		\end{cases}
		$
		\item
		$
		R_0^{0}(L_m)=0
		,\;\;\forall m\in \mathbb{Z}$.
		\item
		$
		R_0^{>1}(L_m)=
		\begin{cases}
		-L_m & m\leqslant1;\\
		0 & m\geqslant2.
		\end{cases}
		$
		\item
		$
		R_0^{<-1}(L_m)=
		\begin{cases}
		-L_m & m\geqslant-1;\\
		0 & m\leqslant-2.
		\end{cases}
		$
		\item
		$
		R_0^{\varnothing}(L_m)=-L_m
		,\;\;\forall m\in \mathbb{Z}$.
		\item
		$
		R_0^{+,\alpha}(L_m)=
		\begin{cases}
		-L_m & m<0; \\
		\alpha L_0 & m=0 ;\\
		0 & m>0,
		\end{cases}
		$\\
		where $\alpha\in \mathbb{C}$.
		\item
		$
		R_0^{-,\alpha}(L_m)=
		\begin{cases}
		-L_m & m>0; \\
		\alpha L_0 & m=0;\\
		0 & m<0,
		\end{cases}\\
		$\\ where $\alpha\in\mathbb{C}$.
	\end{enumerate}
	Conversely, the above operators are all the homogeneous
	Rota-Baxter operators of weight $1$ with degree $0$
	on the Witt algebra $W$.
\end{theorem}

\begin{remark}
	In the above notation, the first part of the superscript represents the zero set of the corresponding $\mathbb{C}$-valued function $f$ and the second part, if it exists, is the value $f(0)$. We use $R_0^0$ instead of $R_0^{\mathbb{Z}}$ since it is the same operator as $R_0^0$ in \autoref{Thm2.14}.
\end{remark}

\begin{remark}\label{rek:RB1}
It is known that $R$ is a Rota-Baxter operator of weight $1$ on a Lie algebra $\frak g$ if and only if so is $-R-{\rm Id}$ on $\frak g$, where
${\rm Id}$ is the identity map on $\frak g$. Using this, we have the following correspondences for the Rota-Baxter operators listed in Theorem~\ref{Thm2.21}:
\[
R_0^{\leqslant1} \;\Longleftrightarrow\;
R_0^{>1},\;\;\;
R_0^{\geqslant-1} \;\Longleftrightarrow\;
R_0^{<-1},\;\;\;
R_0^{0} = 0 \;\Longleftrightarrow\;
R_0^{\varnothing} = -{\rm Id}_W,\;\;\;
R_0^{+,\alpha} \;\Longleftrightarrow\;
R_0^{-,-\alpha-1}.
\]

\end{remark}

\section{Homogeneous Rota-Baxter operators on the Virasoro algebra}

\begin{definition}
The {\it Virasoro algebra} $V$ is a Lie algebra with the basis
$\{L_m,C|m\in\mathbb{Z}\}$ satisfying the following relations:
\begin{equation}\label{Virasoro_1}
    [L_m,L_n]=(m-n)L_{m+n}+\tfrac{m^3-m}{12}\delta_{m+n,0}C,  \forall m,n \in \mathbb{Z}.
\end{equation}
\end{definition}
\begin{equation}\label{Virasoro_2}
[L_m,C]=0,\;\;\forall m\in\mathbb{Z}.
\end{equation}

The Virasoro algebra $V$ is a central extension of the Witt algebra $W$,
and has a natural $\mathbb{Z}$-grading as well:
\begin{align*}
V=\bigoplus_{n\in \mathbb{Z}} V_n,
\end{align*}
where $V_n=\mathbb{C}L_n$ for $n\in \mathbb{Z}^{\ast}$ and $V _0=\mathbb{C}L_0\oplus \mathbb{C}C$.

\begin{remark}\label{W and V} 
	The Witt algebra $W$ is a graded quotient of the Virasoro algebra $V$.
		Any linear graded operator on $W$ can be lifted to that of $V$ by mapping $\mathbb CC$ to $0$.
	Conversely, any linear graded operator $F$ on $V$ can be restricted to a linear operator on $W$ by forgetting the image of $\mathbb CC$.
	If the kernel of $F$ contains the center, then the two can be identified.
\end{remark}

\begin{definition}\label{homogeneous R-B on Vir}
	Let $k$ be an integer. A {\it homogeneous operator $F$ with degree $k$} on the Virasoro algebra $V$ is a linear operator on $V$ satisfying
	\[
		F(V_m)\subset V_{m+k},\;\;
		\forall m\in\mathbb{Z}.
	\]
\end{definition}

Hence {\it homogeneous Rota-Baxter operator $R_k$ with degree $k$} on the Virasoro algebra $V$ is a Rota-Baxter operator on $V$ with the following form:
\begin{equation}\label{eq:RBV1}
R_k(L_m)=f(m+k)L_{m+k}+\theta\delta_{m+k,0}C,\;\;\forall m\in\mathbb{Z};
\end{equation}
\begin{equation}\label{eq:RBV2}
R_k(C)=\mu L_k+\nu\delta_{k,0}C,
\end{equation}
where $f$ is a $\mathbb C$-valued function defined on $\mathbb{Z}$ and $\theta, \mu, \nu\in \mathbb C$.

\subsection{Homogeneous Rota-Baxter operators of weight 0  on the Virasoro algebra}

We begin with the following general result.

\begin{theorem}\label{Thm3.3}
A homogeneous Rota-Baxter operator $R_0$ of weight 0 with degree $0$ on the Virasoro algebra $V$ must be of the form
\begin{align*}
    R_0^{\alpha,\theta,\mu,\nu}(L_m)&=\delta_{m,0}(\alpha L_m+\theta C),\;\;\forall m\in \mathbb{Z},\\
    R_0^{\alpha,\theta,\mu,\nu}(C)&=\mu L_0+\nu C,
  \end{align*}
where $\alpha, \theta, \mu, \nu \in \mathbb C$ are arbitrary. Conversely, the above operators are all the homogeneous
 Rota-Baxter operators of weight $0$ with degree $0$ on the Virasoro algebra $V$.
\end{theorem}
\begin{proof}
    Let $R_0$ be a homogeneous Rota-Baxter operator of weight 0 with degree $0$ on $V$ satisfying Eqs.~(\ref{eq:RBV1}) and (\ref{eq:RBV2}).
    By Eqs. (\ref{RB-lambda}), \eqref{Virasoro_1} and \eqref{Virasoro_2}, we have the following equations:
\begin{equation}\label{V-0-1}
      f(m)f(n)(m-n) = (f(m)+f(n))(m-n)f(m+n)
    \end{equation}
    for any $m+n\neq0$,
    \begin{equation}
      2mf(m)f(-m) = (f(m)+f(-m))\left(2mf(0)+\tfrac{m^3-m}{12}\mu\right),\;\;\forall m\in \mathbb{Z},
    \end{equation}
    and
    \begin{equation}
      f(m)f(-m)\tfrac{m^3-m}{12} = (f(m)+f(-m))(2m\theta + \tfrac{m^3-m}{12} \nu),\;\;\forall m\in \mathbb{Z}.
    \end{equation}
    Let $n=0$ in Eq.~\eqref{V-0-1}. Then $f(m)=0$ for any $m\neq0$. Write $f(0)$ as $\alpha$. Then, all the above equations hold automatically for  arbitrary $\alpha, \theta, \mu, \nu \in\mathbb C$.

Conversely, it is straightforward to check that these
operators are homogeneous Rota-Baxter operator of weight $0$ with degree $0$ on the Virasoro algebra $V$.
  \end{proof}

\begin{remark}\label{rek:RB01 on V}
In view of Remark~\ref{rek:RB0}, a complete set of representatives of the
homogeneous Rota-Baxter operators of weight 0 with degree $0$ on $V$ under the action of $\mathbb{C}^{\ast}$ by scalar multiplication consists of the following operators:
\begin{enumerate}
      \item
      $ R_0^{0,\theta,\mu,\nu}(L_m)=\theta\delta_{m,0} C,\;\;\forall m\in \mathbb{Z}$, and
      $ R_0^{0,\theta,\mu,\nu}(C)=\mu L_0+\nu C$, where $\theta, \mu, \nu \in \mathbb C$ are arbitrary;
      \item
      $R_0^{1,\theta,\mu,\nu}(L_m)=\delta_{m,0}(L_m+\theta C),\;\;\forall m\in \mathbb{Z}$, and
      $R_0^{1,\theta,\mu,\nu}(C)=\mu L_0+\nu C$,
      where $\theta, \mu, \nu \in \mathbb C$ are arbitrary.
\end{enumerate}
\end{remark}

\begin{remark}\label{EndV0}
	Let $E_{0,0}, E_{c,0}, E_{0,c}, E_{c,c}$ be the standard basis of ${\rm End}(V_0)$ given by
	\begin{align*}
	&
	E_{0,0}(L_0) = L_0,\;\;E_{c,0}(L_0) = C,\;\;
	E_{0,c}(L_0) = 0,\;\;E_{c,c}(L_0) = 0,\\
	&
	E_{0,0}(C) = 0,\;\;	E_{c,0}(C) = 0,\;\;
	E_{0,c}(C) = L_0,\;\; E_{c,c}(C) = C.
	\end{align*}
	Extending these linear operators of $V_0$ to those of $V$ by mapping $V_n (n\neq 0)$ to $0$, we see that
	\[
	R_0^{1,0,0,0} = E_{0,0},\;\;
	R_0^{0,1,0,0} = E_{c,0},\;\;
	R_0^{0,0,1,0} = E_{0,c},\;\;
	R_0^{0,0,0,1} = E_{c,c}.
	\]
	Then it is clear that
	\[
	R_0^{\alpha,\theta,\mu,\nu} = \alpha E_{0,0} + \theta E_{c,0} + \mu E_{0,c} + \nu E_{c,c},\;\;\forall \alpha, \theta, \mu, \nu \in \mathbb C.
	\]
	This partly explains our notation in \autoref{Thm3.3}. In this way
	the homogeneous Rota-Baxter operators of weight 0 with degree $0$ on the Virasoro algebra $V$ can be identified
with the linear operators on $V_0$.
\end{remark}

Let $R_k$ be a homogeneous Rota-Baxter operator of weight 0 with a nonzero degree $k$ on $V$
satisfying Eqs.~(\ref{eq:RBV1}) and (\ref{eq:RBV2}). In this case, it is obvious that $\nu=0$, that is,
\begin{align*}
    R_k(L_n)&=f(n+k)L_{n+k}+\theta\delta_{n+k,0}C,\;\;\forall n\in \mathbb{Z};\\
    R_k(C)&=\mu L_k.
  \end{align*}
By Eqs. (\ref{RB-lambda}), \eqref{Virasoro_1} and \eqref{Virasoro_2}, we have the following equations:
  \begin{align}
            &f(m+k)\mu\left(mL_{m+2k}+\tfrac{(m+k)^3-(m+k)}{12}\delta_{m+2k,0}C\right) \label{V-o1}\\
      =~&\mu(m-k)\left(f(m+2k)L_{m+2k}+\theta\delta_{m+2k,0}C\right) +\mu^2\tfrac{m^3-m}{12}\delta_{m+k,0}L_k,\;\;\forall m\in\mathbb{Z}; \nonumber\\
            &f(m)f(n)\left((m-n)L_{m+n}+\tfrac{m^3-m}{12}\delta_{m+n,0}C\right) \label{V-o2}\\
      =~&(f(m)(m-n+k)+f(n)(m-n-k))(f(m+n)L_{m+n}+\theta\delta_{m+n,0}C) \nonumber\\
      +~&\left(\tfrac{m^3-m}{12}f(m)-\tfrac{n^3-n}{12}f(n)\right)\delta_{m+n,k}\mu L_k\nonumber,\;\;\forall m,n\in\mathbb{Z}.
    \end{align}
\begin{proposition}\label{prop:VV1}
With the notations as above, if $\mu=0$, then $f$ and $\theta$ belong to one of the following cases:
    \begin{enumerate}
      \item $f(m)=0$ and $\theta\in\mathbb{C}$;
      \item $f(m)=\alpha\delta_{m+k,0}$,  where $\alpha\in\mathbb{C}$, and $\theta=0$;
      \item $f(m)=\alpha\delta_{m,0}+2\alpha\delta_{m,-\frac{k}{2}}$,  where $\alpha\in\mathbb{C}$, $k$ is even and $\theta\in\mathbb{C}$.
    \end{enumerate}
\end{proposition}
\begin{proof}
If $\mu=0$, then Eq.~\eqref{V-o1} holds automatically and Eq.~\eqref{V-o2} becomes
 \begin{equation}\label{V-1-1}
    f(m)f(n)(m-n)=f(m+n)(f(m)(m-n+k)+f(n)(m-n-k)),\;\;\forall m,n\in\mathbb{Z},
  \end{equation}
  and
  \begin{equation}\label{V-1-2}
    \tfrac{m^3-m}{12}f(m)f(-m)=\theta(f(m)(2m+k)+f(-m)(2m-k)),\;\;\forall m\in \mathbb{Z}.
  \end{equation}

  Note that Eq.~\eqref{V-1-1} is exactly Eq.~\eqref{weight 0}. By the discussion in the previous section, $f$ satisfies one of the following equations:
\begin{enumerate}
\item[(i)]
$
        f(m)=\delta_{m+k,0}f(-k),\;\forall m\in \mathbb{Z};
$
\item[(ii)]
$
      f(m)=(\delta_{m,0}+2\delta_{m,-\frac{k}{2}})f(0),\;\forall m\in\mathbb{Z},
$
 if $k$ is even;
\item[(iii)]
$
       f(m)=
      \sum\limits_{n\in\mathbb{Z}}
      \tfrac{k}{m+k}\delta_{m,ln}f(0),
$
where $l\in\mathbb{Z}^{\ast}$ and $l\nmid k$.
\end{enumerate}

For (i),  Eq.~\eqref{V-1-2} implies that either $\theta=0$ or $f(-k)=0$, which corresponds to cases (2) and (1) respectively.

For (ii),  Eq.~\eqref{V-1-2} holds automatically. Thus $\theta\in \mathbb C$ is  arbitrary. It corresponds to case (3).

For (iii), it does not satisfy Eq.~\eqref{V-1-2}.
\end{proof}

\begin{proposition} \label{prop:VV2}
With the notations as above, if $\mu\neq0$, then
\begin{equation}
  f(m)= -\tfrac{k^2-1}{24}\mu\delta_{m,k},\;\;\forall m\in \mathbb{Z},\end{equation} and $\theta=0$.
\end{proposition}
\begin{proof}
In this case, Eq.~\eqref{V-o1} implies the following equations:
  \begin{equation}\label{V-2-1}
    (m-k)f(m)=(m-2k)f(m+k),\;\forall m\neq0,
  \end{equation}
  \begin{equation}\label{V-2-2}
    \tfrac{1}{2}f(0)=f(k)+\tfrac{k^2-1}{24}\mu
  \end{equation}
  and
  \begin{equation}\label{V-2-3}
    \tfrac{k^2-1}{12}f(-k) = 3\theta.
  \end{equation}
Eq.~\eqref{V-o2} implies the following equations:
  \begin{equation}\label{V-1}
    f(m)f(n)(m-n)=f(m+n)(f(m)(m-n+k)+f(n)(m-n-k))
  \end{equation}
  for $m+n\neq k$,
  \begin{equation}\label{V-2}
    f(m)f(n)(m-n)=~f(k)(2mf(m)-2nf(n))+\mu\left(f(m)\tfrac{m^3-m}{12}-f(n)\tfrac{n^3-n}{12}\right)
  \end{equation}
  for $m+n=k$ and
  \begin{equation}\label{V-3}
    \tfrac{m^3-m}{12}f(m)f(-m)=\theta(f(m)(2m+k)+f(-m)(2m-k)),\;\;\forall m\in \mathbb{Z}.
  \end{equation}

  Let $n=0$ in Eq.~\eqref{V-1}. Then we have
  \begin{equation*}
    f(m)((m+k)f(m)-kf(0))=0,\;\;\forall m\neq k.
  \end{equation*}
 Hence for any  $m\neq k,-k$,  either $f(m)=0$ or
  \begin{equation*}
    f(m)=\tfrac{k}{m+k}f(0).
  \end{equation*}
In addition, $f(-k)f(0)=0$.

  On the other hand, let $m=-k$ in Eq.~\eqref{V-2-1}. Then we have
  \begin{equation*}
    2f(-k)=3f(0).
  \end{equation*}
Hence $f(0)=0$. Therefore $f(m)=0$ for any $m\neq k$. Thus Eqs. \eqref{V-2-1}, \eqref{V-2-2}, \eqref{V-2-3}, \eqref{V-1}, \eqref{V-2} and \eqref{V-3} imply that
$$   f(k) = -\tfrac{k^2-1}{24}\mu,\;\; \theta  = 0.$$
Therefore the conclusion holds.
\end{proof}

Summarizing Propositions~\ref{prop:VV1} and \ref{prop:VV2}, and with a similar discussion as that of
Theorem~\ref{Thm2.14}, we obtain the following classification.

\begin{theorem}\label{Thm3.6}
	A homogeneous Rota-Baxter operator of weight $0$ with a
	nonzero degree on the Virasoro algebra $V$ must be one of the following forms.
	\begin{enumerate}
		\item[(I)]
		$
		R_k^{\theta c}(L_m)=
		\theta\delta_{m+k,0}C
		,\;\;\forall m\in \mathbb{Z}$ and
		$R_k^{\theta c}(C)=0$,
		where $k\in \mathbb{Z}^{\ast}$ and $\theta \in\mathbb{C}$.
		\item[(II)]
		$
		R_k^{\alpha}(L_m)=\alpha\delta_{m+2k,0}L_{m+k},\;\;\forall m\in \mathbb{Z}$ and
		$R_k^{\alpha}(C)=0$, where $k\in \mathbb{Z}^{\ast}$ and $\alpha \in\mathbb{C}^{\ast}$
		\item[(III)]
		$
		R_{2k}^{'\beta+\vartheta c}(L_m)=
		(\beta\delta_{m+2k,0}+2\beta\delta_{m+3k,0})L_{m+2k}+\vartheta\delta_{m+2k,0}C,\;\;\forall m\in\mathbb{Z}$ and
		$R_{2k}^{'\beta+\vartheta c}(C)=0$, where
		$k\in \mathbb{Z}^{\ast}$, $\beta\in\mathbb{C}^{\ast}$ and $\vartheta\in\mathbb{C}$.
		\item[(IV)]
		$
		R_k^{\to\mu}(L_m)=
		-\tfrac{k^2-1}{24}\mu\delta_{m,0} L_{m+k}
		,\;\;\forall m\in\mathbb{Z}$, and $R_k^{\to\mu}(C)=\mu L_k$,
		where $\mu\in\mathbb{C}^{\ast}$.
	\end{enumerate}
	Moreover,
	\begin{enumerate}
		\item If $k\neq 0$ and is odd, then
		$\left\{R_k^{\theta c}, R_k^{\alpha}, R_k^{\to\mu}
		\middle|\alpha,\mu\in\mathbb{C}^{\ast}, \theta\in\mathbb{C}
		\right\}$
		are all the homogeneous Rota-Baxter
		operators of weight $0$ with degree $k$ on the Virasoro algebra $V$.
		\item If $k\neq 0$ and is even, then
		$\left\{R_k^{\theta c}, R_k^{\alpha}, R_{k}^{'\beta+\vartheta c}, R_k^{\to\mu}
		\middle|\alpha,\beta,\mu\in\mathbb{C}^{\ast}, \theta,\vartheta\in\mathbb{C}
		\right\}$
		are all the homogeneous Rota-Baxter
		operators of weight $0$ with degree $k$ on the Virasoro algebra $V$.
	\end{enumerate}
\end{theorem}

\begin{remark} We remark that the symbol $c$ in case (I) is merely formal, it is treated as a vector over $\mathbb{C}$. For this reason, we write $R_k^{0}, R_k^{c}, R_{2k}^{'\beta}$ and $R_{2k}^{'\beta+c}$ instead of $R_k^{0c}, R_k^{1c}, R_{2k}^{'\beta+0c}$ and $R_{2k}^{'\beta+1c}$ respectively. This creates no confusion since if one allows $\alpha=0$ in case (II), then the operator $R_k^0$ agrees with $R_k^{0c}$. It is straightforward to verify that
	\[
	R_k^{\theta c} = \theta R_k^{c},\;\;
	R_k^{\alpha} = \alpha R_k^{1},\;\;
	R_{2k}^{'\beta+\vartheta c} = \beta R_{2k}^{'1} + \vartheta R_{2k}^{c},\;\;
	R_k^{\to\mu} = \mu R_k^{\to1},
	\]
	for any $\alpha,\beta,\mu\in\mathbb{C}^{\ast}$ and $\theta,\vartheta\in\mathbb{C}$.
	
Also note that although
$R_k^{\alpha}$ and $R_{2k}^{'\beta}$ are different from the operators in Theorem~\ref{Thm2.14}
, the operators $R_k^{\alpha}, R_{2k}^{'\beta}$ can be obtained from those in Theorem~\ref{Thm2.14} by obvious extension. 
Therefore, the operators $R_k^{\alpha}, R_{2k}^{'\beta}$ in the two theorems coincide respectively in the sense of Remark~\ref{W and V}.
\end{remark}

\begin{remark}\label{Comparing}
	Comparing Theorem~\ref{Thm2.14} with
	Theorem~\ref{Thm3.6}, we notice that
	the homogeneous Rota-Baxter operators of
	weight $0$ with a nonzero degree on the Witt algebra $W$ and
	Virasoro algebra $V$ are quite different, despite $V$ is a
	central extension of $W$.
	
	Explicitly, although the operators $R_k^{\alpha}, R_{2k}^{'\beta}$ in \autoref{Thm2.14} coincide with those in \ref{Thm3.3} respectively,
	not all Rota-Baxter operators on both algebras can be related in this way.
	On one hand, the operators $R_k^{l,\gamma}$ in Theorem~\ref{Thm2.14} are inconsistent with the structure of the Virasoro algebra $V$.
	On the other hand, there are more operators ($R_k^{\theta c},R_k^{\to\mu}$ and $R_{2k}^{'\beta+\vartheta c}$ with $\vartheta\neq0$) in Theorem~\ref{Thm3.6} that cannot be restricted to the operators in \autoref{Thm2.14}.
	In summary, all homogeneous
	Rota-Baxter operators of weight 0 with a nonzero degree on
	the Virasoro algebra $V$ cannot be obtained directly from those on
	the Witt algebra $W$ nor the latter can be extended to the former, although there are overlaps in certain
	special cases.
	
	The main reason for this phenomenon is the (non-zero) central element $C$ in $V$.
	In general, let $\mathfrak{e}$ be a central extension of the Lie algebra $\mathfrak{g}$
by one-dimensional center $\mathbb CC$.
	Then the brackets $[\cdot,\cdot]_{\mathfrak{g}}$ and $[\cdot,\cdot]_{\mathfrak{e}}$ are related by
	\[
	[x,y]_{\mathfrak{e}} = [x,y]_{\mathfrak{g}} + \epsilon(x,y)C,\;\;\forall x,y\in \mathfrak{g},
	\]
	where $\epsilon(-,-)$ is a $2$-cocycle on $\mathfrak{g}$.
	Therefore, by Eq.~\eqref{RB-lambda}, a Rota-Baxter operator $R$ on $\mathfrak{g}$ can be lifted to $\mathfrak{e}$ 
	if and only if
	\[
	\epsilon(R(x),R(y))=0,\;\;\forall x,y\in\mathfrak{g}.
	\]
	In our case, for a homogeneous Rota-Baxter operator on $\mathfrak{g}$ defined through the $\mathbb{C}$-valued function $f$, the above condition reads
	\begin{equation}\label{lifting}
		\tfrac{m^3-m}{12}f(m)f(n)\delta_{m+n,0} = 0\,\;\;\forall m,n\in\mathbb{Z}.
	\end{equation}
\end{remark}

\begin{remark}\label{rek:RB0 on V}
Note that, 
in view of Remark~\ref{rek:RB0}, a complete set of representatives of all homogeneous Rota-Baxter operators of weight 0 with a nonzero degree $k$ on $V$ under the action of $\mathbb{C}^{\ast}$ by scalar multiplication
is given by
\begin{itemize}
	\item
	$\mathcal{R}^V_k=
	\left\{R_k^{0}=0, R_k^{c}, R_k^{1}, R_k^{\to1}
	\right\}
	$, if $k$ is odd;
	\item
	$\mathcal{R}^V_k=
	\left\{R_k^{0}=0, R_k^{c}, R_k^{1}, R_{k}^{'1}+\vartheta R_k^{c}, R_k^{\to1}
	\middle|\vartheta\in\mathbb{C}
	\right\}
	$, if $k$ is even.
\end{itemize}
\end{remark}

\subsection{Homogeneous Rota-Baxter operators of weight 1 on the Virasoro algebra}

Let $R_k$ be a homogeneous Rota-Baxter operator of weight 1 with degree $k$ on $V$
satisfying Eqs.~(\ref{eq:RBV1}) and (\ref{eq:RBV2}). By Eqs. (\ref{RB-lambda}), \eqref{Virasoro_1} and \eqref{Virasoro_2}, we have the following equations:
 \begin{align}\label{eq:main1}
            &f(m)f(n)\left((m-n)L_{m+n}+\tfrac{m^3-m}{12}\delta_{m+n,0}C\right)\\
      =~&f(m)(m-n+k)(f(m+n)L_{m+n}+\theta\delta_{m+n,0}C) \nonumber \\
      +~&f(m)\tfrac{m^3-m}{12}\delta_{m+n,k}(\mu L_k+\nu\delta_{k,0}C) \nonumber \\
      +~&f(n)(m-n-k)(f(m+n)L_{m+n}+\theta\delta_{m+n,0}C) \nonumber \\
      -~&f(n)\tfrac{n^3-n}{12}\delta_{m+n,k}(\mu L_k+\nu\delta_{k,0}C) \nonumber \\
      +~&(m-n)(f(m+n-k)L_{m+n-k}+\theta\delta_{m+n,k}C) \nonumber \\
      +~&\tfrac{(m-k)^3-(m-k)}{12}\delta_{m+n,2k}(\mu L_k+\nu\delta_{k,0}C),\;\;\forall m,n\in \mathbb{Z}. \nonumber
   \end{align}
If $k\neq 0$, then by Eq.~(\ref{eq:main1}), we have
\begin{equation*}
      (m-n)f(m+n-k)=0,\;\;\forall m,n\in \mathbb{Z}.
 \end{equation*}
Thus $f(m)=0$ for any $m\in\mathbb{Z}$. In this case, by  Eq.~(\ref{eq:main1}) again, we show that
$\theta=0$, $\nu=0$ and $\mu=0$. Hence, any homogeneous Rota-Baxter operator of weight 1 with a nonzero degree $k$ on $V$
is zero.

Next let $R_0$ be a homogeneous Rota-Baxter operator of weight 1 with degree $0$ on $V$.
Then  Eq.~(\ref{eq:main1}) becomes
\begin{align}
            &f(m)f(n)\left((m-n)L_{m+n}+\tfrac{m^3-m}{12}\delta_{m+n,0}C\right)\\
      =~&(m-n)(f(m+n)L_{m+n}+\theta\delta_{m+n,0}C)(f(m)+f(n)+1)\nonumber \\
      +~&\tfrac{m^3-m}{12}\delta_{m+n,0}(\mu L_0+\nu C)(f(m)+f(n)+1),\;\;\forall m,n\in\mathbb{Z}.\nonumber
    \end{align}
By this equation, we have the following equations:
    \begin{equation}\label{Eq_3.16}
      f(m)f(n)=(f(m)+f(n)+1)f(m+n), \;\forall m\neq n, m+n\neq0;
    \end{equation}
      \begin{equation}\label{Eq_3.17}
      mf(m)f(-m)=(f(m)+f(-m)+1)\left(mf(0)+\tfrac{m^3-m}{24}\mu\right),\;\forall m\in\mathbb{Z};
    \end{equation}
    \begin{equation}\label{Eq_3.18}
      \tfrac{m^3-m}{24}f(m)f(-m)=(m\theta+\tfrac{m^3-m}{24}\nu)(f(m)+f(-m)+1),\;\forall m\in\mathbb{Z}.
    \end{equation}

 Let $n=0$ in Eq.~\eqref{Eq_3.16}.  Then we have
    \begin{equation*}
      (f(m)+1)f(m)=0,\;\;\forall m\ne 0.
    \end{equation*}
  Hence $f(m)=0$ or $-1$ for $m\neq0$.

Set
\[
\mathcal{I}_1=\{m\mid f(m)=0\}, \mathcal{I}_2=\{m\mid f(m)=-1\}.
\]

By Lemma~\ref{lemma2.15}, we have the following conclusion.
\begin{lemma}\label{lemma3.7}
Let $f$ be a $\mathbb C$-valued function defined on $\mathbb{Z}$ satisfying Eq.~(\ref{Eq_3.16}). If $m,n\in \mathbb{Z}$ such that $m\ne n$, $m+n\neq 0$ and for $i=1,2$,
$m, n\in \mathcal{I}_i$, then $m+n\in \mathcal{I}_i$.
\end{lemma}
 Let $m=1$ in Eqs. \eqref{Eq_3.17} and \eqref{Eq_3.18}. Then we have
  \begin{equation}\label{Eq_3.19}
    f(1)f(-1)=(f(1)+f(-1)+1)f(0),
  \end{equation}
  \begin{equation}\label{Eq_3.20}
    (f(1)+f(-1)+1)\theta=0.
  \end{equation}
Therefore, we can divide the situation into four cases:
  \begin{enumerate}
    \item[(i)] $1,-1\in \mathcal{I}_1$;
    \item[(ii)] $1,-1\in \mathcal{I}_2$;
    \item[(iii)] $1\in \mathcal{I}_1, -1\in \mathcal{I}_2$;
    \item[(iv)] $1\in \mathcal{I}_2, -1\in \mathcal{I}_1$.
  \end{enumerate}

\begin{proposition}\label{prop3.8}
    If $1,-1\in \mathcal{I}_1$, then $\theta=0$ and $f,\mu,\nu$ for $R_0$ of weight 1 in Eqs. (\ref{eq:RBV1}-\ref{eq:RBV2})
    belong to one of the following cases:
  \begin{enumerate}
    \item $\mathcal{I}_1=\{m\mid m\leqslant1\}, \mathcal{I}_2=\{m\mid m\geqslant2\}$, and $\mu, \nu\in\mathbb C$ are arbitrary;
    \item $\mathcal{I}_1=\{m\mid m\geqslant-1\}, \mathcal{I}_2=\{m\mid m\leqslant-2\}$, and $\mu, \nu\in\mathbb C$ are arbitrary;
    \item $\mathcal{I}_1=\mathbb{Z}, \mathcal{I}_2=\varnothing$ and $\mu=\nu=0$. In this case, $R_0=0$.
  \end{enumerate}
  \end{proposition}
\begin{proof}
When $1,-1\in \mathcal{I}_1$, Eqs. \eqref{Eq_3.19} and \eqref{Eq_3.20} become
  \begin{equation*}
     f(0)=0,\;\;\theta=0.
  \end{equation*}
Thus $0\in \mathcal{I}_1$. Moreover, there are following three cases:
  \begin{enumerate}
    \item If $2\in \mathcal{I}_2$, then $-2\in \mathcal{I}_1$. Otherwise, if $-2\in \mathcal{I}_2$, then Eq.~\eqref{Eq_3.18} implies
               $\nu=-1$. Thus Eq.~\eqref{Eq_3.18} becomes
\begin{equation*}
      \tfrac{m^3-m}{24}f(m)f(-m)=-\tfrac{m^3-m}{24}(f(m)+f(-m)+1),\;\;\forall m\in \mathbb{Z},
\end{equation*}
which contradicts with the assumption that $1,-1\in \mathcal{I}_1$.  Since $-1,-2\in \mathcal{I}_1$, by Lemma~\ref{lemma3.7} and induction, $\mathcal{I}_1$ contains all negative integers. Therefore, for any $m\geqslant2$, $2-m\in \mathcal{I}_1$. Hence $m\in \mathcal{I}_2$. Otherwise, by Lemma~\ref{lemma3.7}, we have $2=(2-m)+m\in \mathcal{I}_1$ which contradicts with the assumption that $2\in \mathcal{I}_2$. Hence
$$\mathcal{I}_1=\{m\mid m\leqslant1\}, \mathcal{I}_2=\{m\mid m\geqslant2\}.$$
In this case, Eqs. \eqref{Eq_3.17} and \eqref{Eq_3.18} hold automatically for $m\neq0,\pm1$. Thus $\mu$ and $\nu$ are arbitrary.
   \item Similarly, if $-2\in \mathcal{I}_2$, then $\mathcal{I}_1=\{m\mid m\geqslant-1\}, \mathcal{I}_2=\{m\mid m\leqslant-2\}$ and $\mu$ and $\nu$ are arbitrary.
    \item If $2,-2\in \mathcal{I}_1$, then $\mathcal{I}_1=\mathbb{Z}, \mathcal{I}_2=\varnothing$. In this case, Eqs. \eqref{Eq_3.17} and \eqref{Eq_3.18} become
               \begin{equation*}
                 \tfrac{m^3-m}{24}\mu=0,\;\;\forall m\neq 0, \pm 1;\;\;\nu=0.
               \end{equation*}
 Thus $\mu=0$ and hence $R_0=0$.
  \end{enumerate}
Therefore the conclusion holds.
\end{proof}

Similarly, for case (ii) of $-1,1\in \mathcal{I}_2$, we have the following conclusion.
\begin{proposition}\label{prop3.9}
If $1,-1\in \mathcal{I}_2$, then $\theta=0$ and $f,\mu,\nu$ for $R_0$ of weight 1 in Eqs. (\ref{eq:RBV1})-(\ref{eq:RBV2})
belong to one of the following cases:
   \begin{enumerate}
    \item $\mathcal{I}_1=\{m\mid m\geqslant2\}, \mathcal{I}_2=\{m\mid m\leqslant1\}$, and $\mu, \nu\in\mathbb C$ are arbitrary;
    \item $\mathcal{I}_1=\{m\mid m\leqslant-2\}, \mathcal{I}_2=\{m\mid m\geqslant-1\}$, and $\mu, \nu\in\mathbb C$ are arbitrary;
    \item $\mathcal{I}_1=\varnothing, \mathcal{I}_2=\mathbb{Z}$ and $\mu=0, \nu=-1$. In this case, $R_0=-{\rm Id}$.
  \end{enumerate}
\end{proposition}

For cases (iii) and (iv), it is straightforward to get the following conclusions.

 \begin{proposition}\label{prop3.10}
    If $1\in \mathcal{I}_1, -1\in \mathcal{I}_2$, then $\mathbb{Z}_+\subset \mathcal{I}_1, \mathbb{Z}_-\subset \mathcal{I}_2$, and $f(0),\mu,\theta,\nu \in \mathbb C$ are arbitrary.
  \end{proposition}
 \begin{proposition}\label{prop3.11}
    If $1\in \mathcal{I}_2, -1\in \mathcal{I}_1$, then $\mathbb{Z}_+\subset \mathcal{I}_2, \mathbb{Z}_-\subset \mathcal{I}_1$, and $f(0),\mu,\theta,\nu \in \mathbb C$ are arbitrary.
  \end{proposition}

From the above propositions, we see that the function $f$ is independent from $\mu,\theta$ and $\nu$.
Moreover, we obtain the following classification.

\begin{theorem}\label{Thm3.17}
A homogeneous Rota-Baxter operator of weight 1 with degree $0$ on the Virasoro algebra $V$ is of the form
\[
R_0^{\ast,\theta,\mu,\nu} = R_0^{\ast} + \theta E_{c,0} + \mu E_{0,c} + \nu E_{c,c},
\]
where $R_0^{\ast}$ is a homogeneous Rota-Baxter operator of weight $1$ with degree $0$ on the Witt algebra $W$ and $E_{c,0}, E_{0,c}, E_{c,c}$ are the operators on $V_0$ mentioned in Remark~\ref{EndV0}.
Moreover, $R_0^{\ast,\mu,\theta,\nu}$
must be one of the following cases.
 \begin{enumerate}
 	\item
 	$R_0^{\leqslant1,0,\mu,\nu}$, where $\mu,\nu\in\mathbb{C}$.
 	\item
 	$R_0^{\geqslant-1,0,\mu,\nu}$, where $\mu,\nu\in\mathbb{C}$.
 	\item
 	$R_0^{0,0,0,0}=0$.
 	\item
 	$R_0^{>1,0,\mu,\nu}$, where $\mu,\nu\in\mathbb{C}$.
 	\item
 	$R_0^{<-1,0,\mu,\nu}$, where $\mu,\nu\in\mathbb{C}$.
 	\item
 	$R_0^{\varnothing,0,0,-1}=-{\rm Id}_V$.
 	\item
 	$R_0^{+,\alpha,\theta,\mu,\nu}$, where $\alpha,\theta,\mu,\nu\in\mathbb{C}$.
 	\item
 	$R_0^{-,\alpha,\theta,\mu,\nu}$, where $\alpha,\theta,\mu,\nu\in\mathbb{C}$.
  \end{enumerate}
Conversely, the above operators are all the homogeneous 
Rota-Baxter operators of weight $1$ with degree $0$ on the Virasoro
algebra $V$.
\end{theorem}
\begin{proof}
	By Eqs.~\eqref{Eq_3.16} and \eqref{Eq_3.19}, the function $f$ is independent from $\mu,\theta$ and $\nu$.
	Comparing Eqs.~\eqref{Eq2.7} with \eqref{Eq_3.16}, we see that the values of $f(m)$ at nonzero $m$ are already given in Propositions~\ref{prop2.16}, \ref{prop2.17} and	\ref{prop2.18}. Moreover, it is easy to check that $f(0)$ given there also satisfies Eq.~\eqref{Eq_3.19}. This shows the first part of the theorem.
	The second part follows from Propositions~\ref{prop3.8}, \ref{prop3.9},
	\ref{prop3.10} and \ref{prop3.11}.
	The final part can be directly verified.
\end{proof}

\begin{remark} Similar to Remark~\ref{rek:RB1}, we have the following correspondences between $R$ and $-R-{\rm Id}$ for the Rota-Baxter operators listed in Theorem~\ref{Thm3.17}:
\begin{align*}
   R_0^{\leqslant1,0,\mu,\nu} &\Longleftrightarrow\;\;
   R_0^{>1,0,-\mu,-\nu-1},\\
   R_0^{\geqslant-1,0,\mu,\nu} &\Longleftrightarrow\;\;
   R_0^{<-1,0,-\mu,-\nu-1},\\
   R_0^{0,0,0,0}=0&\Longleftrightarrow\;\;
   R_0^{\varnothing,0,0,-1}=-{\rm Id}_V,\\
   R_0^{+,\alpha,\theta,\mu,\nu}&\Longleftrightarrow\;\;
   R_0^{-,-\alpha-1,-\theta,-\mu,-\nu-1}.
\end{align*}
\end{remark}
\begin{remark}\label{Comparing2}
	Comparing Proposition~\ref{prop1} with \autoref{Thm3.3} and \autoref{Thm2.21} with \ref{Thm3.17},
	we notice that
	homogeneous Rota-Baxter operators with degree $0$
	on the Witt algebra $W$ and the Virasoro algebra $V$ are closely related.
	Explicitly, any homogeneous Rota-Baxter operator with degree $0$ on $V$ can be written as
	\[
	R_0^{\ast,\theta,\mu,\nu} = R_0^{\ast} + F,
	\]
	where $R_0^{\ast}$ is a homogeneous Rota-Baxter operator with degree $0$ on $W$ and $F$ is a certain linear operator on $V_0$ determined by $\theta,\mu$ and $\nu$.
	Moreover, except for case (3) of Proposition~\ref{prop3.9},
	all functions in Propositions~\ref{prop3.8}, \ref{prop3.9},
	\ref{prop3.10} and \ref{prop3.11} satisfy Eq.~\ref{lifting},
	we therefore conclude that all homogeneous Rota-Baxter operators with degree $0$ on $W$
	can be lifted to $V$ 
	except for $R_0^{\varnothing}=-{\rm Id}_W$.
\end{remark}

\section{Solutions of the CYBE on $W\ltimes_{{\rm ad}^{*}} W^{\ast}$ and $V\ltimes_{{\rm ad}^{*}} V^{*}$}

First we give some notations. Let $\mathfrak{g}$ be a Lie algebra.
 An element $r=\sum\limits_i a_i\otimes b_i \in \mathfrak{g}\otimes\mathfrak{g}$ is called a solution of the {\it classical Yang-Baxter equation (CYBE)} on $\mathfrak{g}$ if $r$ satisfies
\begin{equation*}
  [r_{12}, r_{13}] + [r_{12}, r_{23}] + [r_{13}, r_{23}] = 0 \text{ in } U(\mathfrak{g}),
\end{equation*}
where $U(\mathfrak{g})$ is the universal enveloping algebra of $\mathfrak{g}$ and
\begin{equation*}
  r_{12} = \sum_i a_i \otimes b_i \otimes 1, r_{13} = \sum_i a_i \otimes 1 \otimes b_i, r_{23} = \sum_i 1 \otimes a_i \otimes b_i.
\end{equation*}
For any $r=\sum\limits_i a_i\otimes b_i$, set
\begin{equation*}
r^{21}=\sum\limits_i b_i\otimes a_i.
\end{equation*}
It is obvious that $r$ is skew-symmetric if and only if $r=-r^{21}$.

Let ${\rm ad}:\frak g\rightarrow \mathfrak{gl}(\frak g)$ be the adjoint representation of $\frak g$
 defined by ${\rm ad}(x)(y)=[x,y]$ for any $x,y\in \frak g$.
Let ${\rm ad}^{\ast}\colon\mathfrak{g}\rightarrow \mathfrak{gl}(\mathfrak{g}^{\ast})$ be the dual representation of the adjoint representation of $\mathfrak{g}$. On the vector space $\mathfrak{g}\oplus\mathfrak{g}^{\ast}$, there is a natural Lie algebra structure (denoted by $\mathfrak{g} \ltimes_{{\rm ad}^{\ast}} \mathfrak g^{\ast}$) given by
\begin{equation}\label{semiproduct}
[x_1+f_1, x_2+f_2]=[x_1,x_2]+{\rm ad}^{\ast} (x_1)f_2 - {\rm ad}^{\ast} (x_2)f_1,\;\forall x_1,x_2\in \mathfrak{g}, f_1, f_2\in \mathfrak{g}^{\ast}.
\end{equation}

A linear map is said to be of \emph{finite rank} if its image has finite dimension.
A linear operator $R$ on $\mathfrak{g}$ of finite rank can be identified as an element in
$
\mathfrak{g}\otimes\mathfrak{g}^{\ast}
\subset
(\mathfrak{g} \ltimes_{{\rm ad}^{\ast}} \mathfrak g^{\ast})
\otimes
(\mathfrak{g} \ltimes_{{\rm ad}^{\ast}} \mathfrak g^{\ast})
$ as follows.
Let $\{e_i\}_{i\in {I}}$ be a basis of $\mathrm{Im} R$,
then for $x\in\mathfrak{g}$,
$R(x)$ can be written as a linear combination of the basis.
In other words, for each $i\in {I}$ there exists a unique linear functional $R_i\in\mathfrak{g}^{\ast}$ such that
\[
  R(x)=\sum_{i\in {I}}R_i(x)e_i,\;\;\forall x\in\mathfrak{g}.
\]
Note that $I$ is finite since $R$ is of finite rank.
Then we have
\begin{equation}
R=\sum_{i\in {I}} e_i\otimes R_i\in
\mathfrak{g}\otimes\mathfrak{g}^{\ast}
\subset(\mathfrak{g} \ltimes_{{\rm ad}^{\ast}} \mathfrak g^{\ast})
\otimes(\mathfrak{g} \ltimes_{{\rm ad}^{\ast}} \mathfrak g^{\ast}).
\end{equation}

\begin{lemma}\label{lem:op} \cite{Bai1}
 Let $\mathfrak{g}$ be a Lie algebra.
 A linear operator $R$ on $\mathfrak{g}$ of finite rank
 is a Rota-Baxter operator of weight $0$ on $\frak g$
 if and only if $r = R - R^{21}$ is a skew-symmetric solution
 of the CYBE on $\mathfrak{g} \ltimes_{{\rm ad}^{\ast}} \mathfrak g^{\ast}$.
\end{lemma}

\begin{remark}
Note that the above conclusion was originally proved for the finite dimensional case and it is not hard to extend it to the infinite dimensional case for linear operators of finite rank.
\end{remark}

For the Witt algebra $W$, let $\{L^{\ast}_n\}_{n\in\mathbb{Z}}$ be the dual basis of $\{L_n\}_{n\in\mathbb{Z}}$. Then the Lie algebra structure on $W\ltimes_{{\rm ad}^{\ast}}W^{\ast}$ is given by
\begin{equation}
  [L_m,L_n] = (m-n)L_{m+n},\;\;[L_m,L^{\ast}_n]  = (n-2m)L^{\ast}_{n-m},\;\;[L^{\ast}_m,L^{\ast}_n]=0,\;\;\forall m,n\in \mathbb{Z}.
\end{equation}

Note that the Rota-Baxter operators of weight 0 on $W$ given in Theorem~\ref{Thm2.14} are of finite rank except for those of type (III).
By Lemma~\ref{lem:op} we obtain the following skew-symmetric solutions of the CYBE on $W\ltimes_{{\rm ad}^{\ast}}W^{\ast}$.
 \begin{enumerate}
      \item[(I)]
      $
      r_k^{\alpha} = \alpha\left(L_{-k}\otimes L^{\ast}_{-2k} - L^{\ast}_{-2k} \otimes L_{-k}\right)
      $,
      where $k\in\mathbb{Z}, \alpha\in\mathbb{C}$;
      \item[(II)]
      $
      r_{2k}^{'\beta}= \beta\left(L_0\otimes L^{\ast}_{-2k} - L^{\ast}_{-2k}\otimes L_0 + 2L_{-k}\otimes L^{\ast}_{-3k} - 2L^{\ast}_{-3k} \otimes L_{-k}\right)
      $,
      where $k\in\mathbb{Z}^{\ast}$ and $\beta\in\mathbb{C}^{\ast}$.
    \end{enumerate}

For the Virasoro algebra $V$, let $\{L^{\ast}_n\}_{n\in\mathbb{Z}}\cup\{C^{\ast}\}$ be the dual basis of $\{L_n\}_{n\in\mathbb{Z}}\cup\{C\}$. Then the Lie algebra structure on $V\ltimes_{{\rm ad}^{\ast}}V^{\ast}$ is given by
\begin{align}
  [L_m,L_n] & = (m-n)L_{m+n}+\tfrac{m^3-m}{12}\delta_{m+n,0}C,\;\; [L_m,L^{\ast}_n]  = (n-2m)L^{\ast}_{n-m},\\
  [L_m,C^{\ast}]& = -\tfrac{m^3-m}{12}L^{\ast}_{-m},
  \;\; [L^{\ast}_m,L_n^{\ast}]=[L^{\ast}_m,C]=[L_m,C]=[C^{\ast}, C] = 0,\;\;\forall m,n\in \mathbb{Z}.\nonumber
\end{align}

{
	Note that the Rota-Baxter operators of weight 0 on $V$ given in Theorems~\ref{Thm3.3} and \ref{Thm3.6} are all of finite rank.
	By Lemma~\ref{lem:op} we obtain the following skew-symmetric solutions of the CYBE on $V\ltimes_{{\rm ad}^{\ast}}V^{\ast}$.
	\begin{enumerate}
		\item[(0)]
		$r_0^{\alpha,\theta,\mu,\nu} = \alpha e_{0,0} + \theta e_{c,0} + \mu e_{0,c} + \nu e_{c,c}$,
		where $\alpha,\theta,\mu,\nu\in\mathbb{C}$ and
		\begin{align*}
			e_{0,0} &= L_0\otimes L^{\ast}_0 - L^{\ast}_0\otimes L_0,\;\;&
			e_{c,0} &= C\otimes L^{\ast}_0 - L^{\ast}_0\otimes C,\\
			e_{0,c} &= L_0\otimes C^{\ast} - C^{\ast}\otimes L_0,\;\;&
			e_{c,c} &= C\otimes C^{\ast} - C^{\ast}\otimes C.
		\end{align*}
		\item[(I)]
		$r_k^{\theta c} = \theta (C\otimes L^{\ast}_{-k} - L^{\ast}_{-k}\otimes C)$, where $k\in\mathbb{Z}^{\ast}, \theta \in\mathbb{C}^{\ast}$ and
		\item[(II)]
		$r_k^{\alpha}$, where $k\in\mathbb{Z}^{\ast},\alpha\in\mathbb{C}^{\ast}$.
		\item[(III)]
		$r_k^{'\beta + \vartheta c} = r_k^{'\beta} + r_k^{\vartheta c}$,
		where $k\in \mathbb{Z}^{\ast}, \beta\in \mathbb{C}^{\ast}, \vartheta\in\mathbb{C}$.
		\item[(IV)]
		$r_k^{\to \mu} = \mu(-\tfrac{k^2-1}{24}(L_{k}\otimes L^{\ast}_0-L^{\ast}_0\otimes L_{k})
		+ L_k\otimes C^{\ast} - C^{\ast}\otimes L_k)$, where $k\in\mathbb{Z}^{\ast}, \mu\in\mathbb{C}^{\ast}$.
	\end{enumerate}
}

\begin{lemma} \cite{BGN2}\label{lem4.2}
  Let $\frak g$ be a finite-dimensional Lie algebra and $R: \mathfrak{g}\rightarrow \mathfrak{g}$ a linear map. Then $R$ is a Rota-Baxter operator of weight 1 on $\frak g$ if and only if both $(R-R^{21}) + {\rm Id}$ and $(R-R^{21}) - {\rm Id}^{21}$ are solutions of the CYBE on $\mathfrak{g}\ltimes_{{\rm ad}^{*}} \mathfrak{g}^{*}$.
\end{lemma}

\begin{remark}
Since $R$ is a Rota-Baxter operator of weight 1 on a Lie algebra $\frak g$ if and only if so is $-R-{\rm Id}$ on $\frak g$ and
\begin{equation*}
  ((-R-{\rm Id})-(-R-{\rm Id})^{21}) + {\rm Id} = -((R-R^{21}) - {\rm Id}^{21}),
\end{equation*}
we only list the solutions of the CYBE obtained from $(R-R^{21}) + {\rm Id}$.
\end{remark}

For infinite dimensional vector spaces $V_1$ and $V_2$, we define the formal tensor product
$V_1\widehat{\otimes}V_2$ to be the space of formal series on the basis of $V_1\otimes V_2$. Its elements are called {\it formal tensors}. A formal tensor can also be identified as an infinite matrix with the basis of $V_1$ as its row-index set and the basis of $V_2$ as its column-index set. We will not distinguish these two
presentations in this paper.

For a Lie algebra $\mathfrak{g}$, the CYBE on $\mathfrak{g}$ is an equation of tensors in $\mathfrak{g}\otimes\mathfrak{g}$. We need to generalize the notion of CYBE to formal tensors with suitable conditions.

Note that for $r=\sum\limits_{i,j\in {I}} a_{ij} e_i\otimes e_j \in \mathfrak{g}{\otimes}\mathfrak{g}$, the CYBE equals to the following equations:
\begin{equation}\label{formal CYBE}
  [[r]](e_i,e_j,e_k):=\sum_{s,t\in {I}}(C^{i}_{st}a_{sj}a_{tk}+a_{is}C^{j}_{st}a_{tk}+a_{is}a_{jt}C^{k}_{st})=0,\;\;\forall i,j,k\in {I},
\end{equation}
where $C^{i}_{rs}$ are the structural coefficients of $\mathfrak{g}$. The summation is finite
since only finitely many coefficients of $r$ are nonzero.

An infinite matrix $(a_{ij})_{i\in {I}, j\in {J}}$
is said to be {\it row-finite} if each row contains
only finitely many nonzero entries.
An infinite matrix is said to be {\it column-finite} if each column contains only finitely many nonzero entries.
For example, a linear map, viewed as an infinite matrix, is column-finite and vice versa.
An infinite matrix which is both row-finite and column-finite is said to be {\it row-and-column-finite}.

For a formal tensor $r=\sum\limits_{i,j\in {I}} a_{ij} e_i\otimes e_j \in \mathfrak{g}\widehat{\otimes}\mathfrak{g}$, to ensure the summation in Eq.~(\ref{formal CYBE}) is finite, we need $(a_{ij})_{i,j\in {I}}$ to be a row-and-column-finite matrix.

Therefore, a formal tensor $r=\sum\limits_{i,j\in {I}} a_{ij} e_i\otimes e_j \in \mathfrak{g}\widehat{\otimes}\mathfrak{g}$ is called a solution of the \emph{formal CYBE} if it is row-and-column-finite and satisfies Eq.~(\ref{formal CYBE}).

A linear operator $R$ on $\mathfrak{g}$ can be identified as an element in
$\mathfrak{g}\widehat{\otimes}\mathfrak{g}^{\ast}
\subset(\mathfrak{g} \ltimes_{{\rm ad}^{\ast}} \mathfrak g^{\ast})
\widehat{\otimes}(\mathfrak{g} \ltimes_{{\rm ad}^{\ast}} \mathfrak g^{\ast})$
as follows.
Let $\{e_i\}_{i\in {I}}$ be a basis of $\mathfrak{g}$ and
$\{e^{\ast}_i\}_{i\in {I}}$ be its dual defined by
\begin{equation*}
  e_i^{\ast}(e_j)=\delta_{ij},\;\;\forall i,j\in {I}.
\end{equation*}
By Zorn's lemma, $\{e^{\ast}_i\}_{i\in {I}}$ can be extended to a basis of $\mathfrak{g}^{\ast}$, say $\{e^{\ast}_i\}_{i\in {I}}\cup\{f_j\}_{j\in {J}}$. Then we have
\begin{equation*}
  R=\sum_{i\in {I}} R(e_i)\otimes e^{\ast}_i +\sum_{j\in {J}}0\otimes f_j
  \in \mathfrak{g}\widehat{\otimes}\mathfrak{g}^{\ast}
  \subset(\mathfrak{g} \ltimes_{{\rm ad}^{\ast}} \mathfrak g^{\ast})
  \widehat{\otimes}(\mathfrak{g} \ltimes_{{\rm ad}^{\ast}} \mathfrak g^{\ast}).
\end{equation*}

If $R^{21}$ is also column-finite, then we say $R$ is {\it balanced}.
Both Lemma~\ref{lem:op} and \ref{lem4.2} can be easily extended to the infinite dimensional case for balanced Rota-Baxter operators. Therefore we have the following conclusion.

\begin{lemma}\label{lem4.5}
  Let $\mathfrak{g}$ be a Lie algebra and $R\colon \mathfrak{g}\rightarrow\mathfrak{g}$ a balanced linear map. Then
  \begin{enumerate}
    \item $R$ is a Rota-Baxter operator of weight $0$ on $\frak g$ if and only if $r = R - R^{21}$ is a skew-symmetric solution of the formal CYBE on $\mathfrak{g} \ltimes_{{\rm ad}^{\ast}} \mathfrak g^{\ast}$. In particular, when $R$ is of finite rank, it coincides with the conclusion in Lemma~\ref{lem:op}.
    \item $R$ is a Rota-Baxter operator of weight 1 on $\frak g$ if and only if both $(R-R^{21}) + {\rm Id}$ and $(R-R^{21}) - {\rm Id}^{21}$ are solutions of the formal CYBE on $\mathfrak{g}\ltimes_{{\rm ad}^{*}} \mathfrak{g}^{*}$.
  \end{enumerate}
\end{lemma}

Indeed, by Eq.~(\ref{semiproduct}), one sees that Eq.~(\ref{formal CYBE}) is trivial except for the cases: $[[r]](e_i^{\ast},e_j^{\ast},e_k)$, $[[r]](e_i^{\ast},e_j,e_k^{\ast})$ and $[[r]](e_i,e_j^{\ast},e_k^{\ast})$ for $i,j,k\in {I}$. However, for $r=(R-R^{21})+{\rm Id}$ (resp. $r=R-R^{21}$), these equations are nothing but Eq.~(\ref{RB-lambda}) with $\lambda=1$ (resp. $\lambda=0$) and $x=e_i,y=e_j; x=e_i,y=e_k; x=e_j,y=e_k$ respectively.

Therefore, except for the solutions of CYBE obtained from the Rota-Baxter operators of weight $0$ on $W$ which are of finite rank,
 Lemma~\ref{lem4.5} (1) gives the following solutions of the formal CYBE on $W \ltimes_{{\rm ad}^{\ast}} W^{\ast}$ from the Rota-Baxter operators of weight 0 on $W$ which are of type (III) given in Theorem~\ref{Thm2.14}:
   \begin{equation*}
      r_k^{l,\gamma} = \sum\limits_{m\equiv -k \pmod{l}}\tfrac{k}{m+2k}\gamma(L_{m+k}\otimes L^{\ast}_m - L^{\ast}_m\otimes L_{m+k}),
   \end{equation*}
      where $k,l\in\mathbb{Z}^{\ast}$, $l\nmid k$ and $\gamma\in\mathbb{C}^{\ast}$.

  Moreover, Lemma~\ref{lem4.5} (2) gives the following solutions of the formal CYBE on $W \ltimes_{{\rm ad}^{\ast}} W^{\ast}$ from the Rota-Baxter operators of weight 1 on $W$ given in Theorem~\ref{Thm2.21}.
  Note that here ${\rm Id}_W=\sum\limits_{m\in\mathbb{Z}}L_m\otimes L_m^{\ast}$.
   \begin{enumerate}
      \item
      $
       r_0^{\leqslant1} = \sum\limits_{m\leqslant1}L_m\otimes L^{\ast}_m + \sum\limits_{m>1}L^{\ast}_m\otimes L_m
      $.
      \item
      $
       r_0^{\geqslant-1} = \sum\limits_{m\geqslant-1}L_m\otimes L^{\ast}_m + \sum\limits_{m<-1}L^{\ast}_m\otimes L_m
      $.
      \item
      $
        r_0^0 = \sum\limits_m L_m\otimes L^{\ast}_m = {\rm Id}_W
      $.
      \item
      $
        r_0^{>1} = \sum\limits_{m>1}L_m\otimes L^{\ast}_m + \sum\limits_{m\leqslant1}L^{\ast}_m\otimes L_m
      $.
      \item
      $
        r_0^{<-1} = \sum\limits_{m<-1}L_m\otimes L^{\ast}_m + \sum\limits_{m\geqslant-1}L^{\ast}_m\otimes L_m
      $.
      \item
      $
        r_0^{\varnothing} = \sum\limits_m L^{\ast}_m\otimes L_m = {\rm Id}_W^{21}
      $.
      \item
      $
        r_0^{+,\alpha} = \sum\limits_{m<0}L^{\ast}_m\otimes L_m  +\sum\limits_{m>0}L_m\otimes L^{\ast}_m + (\alpha+1)L_0\otimes L^{\ast}_0 - \alpha L^{\ast}_0\otimes L_0
      $, where $\alpha\in\mathbb{C}$.
      \item
      $
        r_0^{-,\alpha} = \sum\limits_{m>0}L^{\ast}_m\otimes L_m  +\sum\limits_{m<0}L_m\otimes L^{\ast}_m + (\alpha+1)L_0\otimes L^{\ast}_0 - \alpha L^{\ast}_0\otimes L_0
      $, where $\alpha\in\mathbb{C}$.
    \end{enumerate}

    We remark that
    although the summation is infinite, the solution of formal CYBE on any highest weight irreducible representation of $W$ will be finite.

Lemma~\ref{lem4.5} (2) also gives the following solutions of the formal CYBE on $V \ltimes_{{\rm ad}^{\ast}} V^{\ast}$ from the Rota-Baxter operators of weight 1 on $V$ given in Theorem~\ref{Thm3.17}.
Note that here ${\rm Id}_V=\sum\limits_{m\in\mathbb{Z}}L_m\otimes L_m^{\ast}+C\otimes C^{\ast}$.
    \begin{enumerate}
      \item
      $r_0^{\leqslant1} + \mu e_{0,c} + \nu e_{c,c} + C\otimes C^{\ast}$,
      where $\mu,\nu\in\mathbb{C}$.
      \item
      $r_0^{\geqslant-1} + \mu e_{0,c} + \nu e_{c,c} + C\otimes C^{\ast}$,
      where $\mu,\nu\in\mathbb{C}$.
      \item
      $r_0^0 + C\otimes C^{\ast}$ i.e. ${\rm Id}_V$.
      \item
      $r_0^{>1} + \mu e_{0,c} + \nu e_{c,c} + C\otimes C^{\ast}$,
      where $\mu,\nu\in\mathbb{C}$.
      \item
      $r_0^{<-1} + \mu e_{0,c} + \nu e_{c,c} + C\otimes C^{\ast}$,
      where $\mu,\nu\in\mathbb{C}$.
      \item
      $r_0^{\varnothing} + C^{\ast}\otimes C$ i.e. ${\rm Id}_V^{21}$.
     \item
     $r_0^{+,\alpha} + \theta e_{c,0} + \mu e_{0,c} + \nu e_{c,c} + C\otimes C^{\ast}$, where $\alpha,\theta,\mu,\nu\in\mathbb{C}$.
      \item
      $r_0^{-,\alpha} + \theta e_{c,0} + \mu e_{0,c} + \nu e_{c,c} + C\otimes C^{\ast}$, where $\alpha,\theta,\mu,\nu\in\mathbb{C}$.
     \end{enumerate}

    We remark that
    although the summation is infinite, the solution of the formal CYBE on any highest weight irreducible representation
    of $V $ will be a finite expression.

\section{Induced pre-Lie algebras on the Witt and Virasoro algebras}
\begin{definition}
A {\it pre-Lie algebra} $A$ is a vector space $A$ with a bilinear product $\ast$ satisfying
  \begin{equation}\label{eq:prelie}
    (x\ast y)\ast z - x\ast(y\ast z) = (y\ast x)\ast z - y\ast(x\ast z),\;\;\forall x,y,z\in A.
      \end{equation}
\end{definition}

\begin{proposition}\label{prop4.2} \cite{GS}
 Let $(\mathfrak{g},[\cdot, \cdot])$ be a Lie algebra with a Rota-Baxter operator $R$ of weight 0 on it. Define a new operation $x\ast y = [R(x),y]$ for any $x,y\in A$. Then $(\mathfrak{g},\ast)$ is a pre-Lie algebra.
\end{proposition}

 Let $(A, \ast)$ be a pre-Lie algebra.  Then the commutator
  \begin{equation}
    [x, y]_{\ast} = x\ast y - y\ast x,\;\;\forall x,y\in A
  \end{equation}
  defines a Lie algebra $\mathfrak{g}(A)$ called the {\it sub-adjacent Lie algebra} of $A$ and $A$ is also called a {\it compatible pre-Lie algebra} on the Lie algebra $\mathfrak{g}(A)$.

 Let us denote the commutator of the induced pre-Lie algebra given in Proposition~\ref{prop4.2} by $[\cdot,\cdot]_R$.
 If $R'$ is another Rota-Baxter operator of weight 0 on $\mathfrak{g}$ such that $\alpha R + \beta R'$ is such a Rota-Baxter operator for
 any scalars $\alpha$ and $\beta$, then
  \begin{equation*}
  	[\cdot,\cdot]_{\alpha R + \beta R'} = \alpha[\cdot,\cdot]_R+\beta[\cdot,\cdot]_{R'}.
  \end{equation*}
  Moreover, after rescaling   by $L_m\rightarrow \tfrac{1}{\alpha}L_m$
  for $\alpha\in \mathbb  C^{\ast}$,
  the pre-Lie algebras induced by $\alpha R$ can be identified with that induced by $R$. Therefore, to give all pre-Lie algebras induced from Rota-Baxter operators of weight 0 on the Lie algebra $\mathfrak{g}$, it suffices to consider only a complete set of representatives of those operators under the $\mathbb C^*$-action (cf. Remark~\ref{rek:RB0}).

\subsection{Induced pre-Lie algebras on the Witt algebra}
\mbox{}

Using the construction given in Prop. \ref{prop4.2}, we have the following.
\begin{theorem}\label{Thm4.3}
The homogeneous Rota-Baxter operators of weight 0 on the Witt algebra $W$ obtained in Theorem~\ref{Thm2.14}  give rise to the following pre-Lie algebras on the underlying vector space $W$:
\begin{enumerate}
	\item[(0)]
	$L_m\ast^0 L_n = 0,
	\;\;\forall m,n\in\mathbb{Z}$.
	\item[(I)]
	$L_m \ast_k L_n = (k-n)\delta_{m,0}L_{n-k},
	\;\;\forall m,n\in\mathbb{Z}$,
	where $k\in\mathbb{Z}$.
	\item[(II)]
	$L_m \ast_{2k}' L_n = (2k-n)\delta_{m,0}L_n + (2k-2n)\delta_{m+k,0}L_{n-k},
	\;\;\forall m,n\in\mathbb{Z}$,
	where $k\in\mathbb{Z}^{\ast}$.
	\item[(III)]
	$L_m\ast_k^l L_n =
	\tfrac{(m+k-n)k}{m+k}\delta_{m,l\mathbb{Z}}L_{m+n}
	\;\;\forall m,n\in\mathbb{Z}$,
	where $k,l\in \mathbb{Z}^{\ast}$ and $l\nmid k$.
\end{enumerate}
Moreover, these pre-Lie algebras are not mutually isomorphic.
\end{theorem}

\begin{proof}
The conclusion follows from Proposition~\ref{prop4.2} by direct computation.
Note that the complete set is chosen to be the union $\bigcup_{k\in\mathbb{Z}}\mathcal{R}^W_k$
with $\mathcal{R}^W_k$ as in Remark~\ref{rek:RB0}.
Moreover, we also use a degree shifting by
for
$L_m\rightarrow L_{m+2k}$ in  (I) and (II), and
$L_m\rightarrow L_{m+k}$ in (III) respectively.
It is also straightforward to show that these pre-Lie algebras are not mutually isomorphic.
\end{proof}

The following conclusion is an immediate consequence.

\begin{proposition}
The sub-adjacent Lie algebras of the pre-Lie algebras in  \ref{Thm4.3} are given as follows respectively.
\begin{enumerate}
	\item[(0)]
	$[L_m,L_n]_{0}=0\;\;\forall m,n\in\mathbb{Z}$.
	\item[(I)]
	$
	[L_m,L_n]_{R_k^1}=
	\begin{cases}
	(k-n)L_{n-k} & m=0,n\neq0;\\
	0 & m,n\neq0;
	\end{cases}
	$\\
	where $k\in\mathbb{Z}$.
	\item[(II)]
	$
	[L_m,L_n]_{R_{2k}^{'1}}=
	\begin{cases}
	(2k-n) L_n & m=0,n\neq0,-k;\\
	(2k-2n) L_{n-k} & m=-k,n\neq0,-k;\\
	kL_{-k} & m=0, n=-k;\\
	0 & m, n\neq0, -k,
	\end{cases}
	$\\
	where $k\in\mathbb{Z}^{\ast}$.
	\item[(III)]
	$
	[L_m, L_n]_{R_k^{l,1}}=
	\begin{cases}
	\tfrac{(m+k-n)k}{m+k} L_{m+n} & m\in l\mathbb{Z}, n\notin l\mathbb{Z};\\
	\tfrac{(m-n)(m+n+k)k}{(m+k)(n+k)} L_{m+n} & m, n\in l\mathbb{Z};\\
	0 &  m,n\notin l\mathbb{Z},
	\end{cases}
	$\\
	where $k,l\in \mathbb{Z}^{\ast}$ and $l\nmid k$.
\end{enumerate}
\end{proposition}

\subsection{Induced pre-Lie algebras on the Virasoro algebra}
Similarly in the case of the Virasoro algebra, we can derive the following pre-Lie algebra structures by Prop. \ref{prop4.2}.
\begin{theorem}\label{Thm4.5}
The homogeneous Rota-Baxter operators of weight 0 with degree $0$ on the Virasoro algebra $V$ obtained in Theorem~\ref{Thm3.3}  give rise to the following pre-Lie algebras on the underlying vector space $V$:
\begin{enumerate}
\item $L_m\ast_{(0,0)} L_n=C\ast_{(0,0)} L_m=L_m\ast_{(0,0)} C=0,\;\;\forall m,n\in\mathbb{Z}$.
\item $L_m\ast_{(1,0)} L_n=-n\delta_{m,0} L_n, \;\; C\ast_{(1,0)} L_n=L_n\ast_{(1,0)} C=0,\;\;\forall m,n\in\mathbb{Z}$.
\item $L_m\ast_{(0,1)} L_n=0, \;\; C\ast_{(0,1)} L_n=- n L_n,\;\; L_n\ast_{(0,1)} C=0,\;\;\forall m,n\in\mathbb{Z}$.
\item $L_m\ast_{(1,1)} L_n=-n\delta_{m,0} L_n, \;\; C\ast_{(1,1)} L_n=- n L_n,\;\; L_n\ast_{(1,1)} C=0,\;\;\forall m,n\in\mathbb{Z}$.
\end{enumerate}
Moreover, these pre-Lie algebras are not mutually isomorphic.
\end{theorem}

\begin{proof}
By Proposition~\ref{prop4.2}, the induced pre-Lie algebra from $R_0^{\alpha,\theta,\mu,\nu}$ is given by
\[
L_m\ast_{(\alpha,\mu)} L_n=-n\alpha\delta_{m,0} L_n, \;\; C\ast_{(\alpha,\mu)} L_n=- n\mu L_n,\;\; L_n\ast_{(\alpha,\mu)} C=0,\;\;\forall m,n\in\mathbb{Z},
\]
where $\alpha,\mu\in \mathbb C$. For $\alpha\ne 0$ or $\mu\ne 0$, we use the linear transformation by
$L_m\rightarrow \tfrac{1}{\alpha}L_m$ for any $m\in \mathbb{Z}$ or $C\rightarrow \tfrac{1}{\mu} C$. Then the conclusion follows.
\end{proof}

The following conclusion follows immediately.
\begin{proposition}
The sub-adjacent Lie algebras of the pre-Lie algebras in  Theorem~\ref{Thm4.5} are given as follows respectively.
\begin{enumerate}
\item $[L_m, L_n]_{(0,0)}=[C, L_m]_{(0,0)}=0,\;\;\forall m,n\in\mathbb{Z}$.
\item $[L_m, L_n]_{(1,0)}=-n\delta_{m,0} L_n+m\delta_{n,0} L_m, \;\; [C, L_n]_{(1,0)}=0,\;\;\forall m,n\in\mathbb{Z}$.
\item $[L_m, L_n]_{(0,1)}=0, \;\; [C, L_n]_{(0,1)}=- n L_n,\;\;\forall m,n\in\mathbb{Z}$.
\item $[L_m, L_n]_{(1,1)}=-n\delta_{m,0} L_n+m\delta_{n,0} L_m, \;\; [C, L_n]_{(1,1)}=- n L_n,\;\;\forall m,n\in\mathbb{Z}$.
\end{enumerate}
\end{proposition}

\begin{theorem} \label{Thm4.7}
The homogeneous Rota-Baxter operators of weight 0 with a nonzero degree on the Virasoro algebra $V$ obtained in Theorem~\ref{Thm3.6} provide the following pre-Lie algebras on the underlying vector space $V$:
 \begin{enumerate}
      \item[(I)]
      $L_m\ast^0 L_n=C\ast^0 L_m=L_m\ast^0 C=0,\;\;\forall m,n\in\mathbb{Z}$.
      \item[(II)]
      $L_m\ast_k L_n=\delta_{m,0}\left((k-n) L_{n-k}-\tfrac{k^3-k}{12}\delta_{n,3k}C\right), \;\;
      C\ast_{k} L_n=L_n\ast_{k} C=0,\;\;\forall m,n\in\mathbb{Z}$, where $k\in\mathbb{Z}^{\ast}$.
      \item[(III)]
      $L_m\ast_{2k}' L_n=\delta_{m,0}(2k- n) L_n
      +2\delta_{m+k,0}\left((k-n) L_{n-k}-\tfrac{k^3-k}{12}\delta_{n,3k}C\right),\;\;
      C\ast_{2k}' L_n=L_n\ast_{2k}' C=0,\;\;\forall m,n\in\mathbb{Z}$, where $k\in\mathbb{Z}^{\ast}$.
      \item[(IV)]
      $L_m\ast_k^{\to} L_n=-\tfrac{k^2-1}{24}\delta_{m,0}C\ast_k^{\to}  L_n, \;\;
      C\ast_k^{\to}  L_n=(k-n) L_{n+k}+\tfrac{k^3-k}{12}\delta_{n+k,0}C,
      \;\; L_n\ast_k^{\to} C=0,\;\;\forall m,n\in\mathbb{Z}$, where $k\in\mathbb{Z}^{\ast}$.
    \end{enumerate}
Moreover, these pre-Lie algebras are not mutually isomorphic.
\end{theorem}
\begin{proof}
The conclusion follows from Proposition~\ref{prop4.2} by a direct
computation.
Note that the complete set is chosen to be the union $\bigcup_{k\in\mathbb{Z}^{\ast}}\mathcal{R}^V_k$
with $\mathcal{R}^V_k$ as in Remark~\ref{rek:RB0 on V}
and that the operators $R_k^0$ and $R_k^{\theta c}$ give the same pre-Lie algebra (I).
Moreover, we also use a degree shifting:
$L_m\rightarrow L_{m+2k}$ for (II) and (III) respectively. It is also straightforward to check
that these pre-Lie algebras are not mutually
isomorphic.
\end{proof}

\begin{proposition}
The sub-adjacent Lie algebras of the pre-Lie algebras in  Theorem~\ref{Thm4.7} are given as follows respectively.
\begin{enumerate}
      \item[(I)]
      $[L_m, L_n]_{0}=[C, L_m]_{0}=0,\;\;\forall m,n\in\mathbb{Z}$.

      \item[(II)]  	
      	$[L_m, L_n]_{R_k^1}=
      	\begin{cases}
      	-2k L_{2k}-\tfrac{k^3-k}{12}C & m=0, n=3k;\\
      	(k-n) L_{n-k} & m=0, n\neq0,3k;\\
      	0 & m, n\neq0,
      	\end{cases}$\\
      	$[C, L_n]_{R_k^1}=0,\;\;\forall n\in\mathbb{Z}$,
      	where $k\in\mathbb{Z}^{\ast}$.

      \item[(III)]
      $[L_m, L_n]_{R_{2k}^{'1}}=
      \begin{cases}
      -4k L_{2k}
      +\tfrac{k-k^3}{6}C &
      m=-k, n=3k;\\
      - k L_{-k}&
      m=-k, n=0;\\
      2(k-n)L_{n-k}&
      m=-k,n\neq0,-k,3k;\\
      (2k-n)L_n&
      m=0,n\neq0,-k;\\
      0&
      m,n\neq0,-k,
      \end{cases}$\\
      $[C, L_n]_{R_{2k}^{'1}}=0,\;\;\forall n\in\mathbb{Z}$, where $k\in\mathbb{Z}^{\ast}$.

      \item[(IV)]
      $[L_m, L_n]_{R_k^{\to}}=
      \begin{cases}
      -\tfrac{k^2-1}{24}\left(2kL_0+\tfrac{k^3-k}{12}C\right) &
      m=0, n=-k; \\
      -\tfrac{k^2-1}{24}(k-n) L_{n+k}&
      m=0, n\neq0,-k;\\
      0&
      m,n\neq0,
      \end{cases}$\\
      $[C,L_n]_{R_k^{\to}}=
      \begin{cases}
      2kL_0+\tfrac{k^3-k}{12}C &
      n=-k;\\
      (k-n)L_{n+k}&
      n\neq-k,
      \end{cases}$\\
      where $k\in\mathbb{Z}^{\ast}$.
\end{enumerate}

\end{proposition}
\section{Induced PostLie algebras on the Witt and Virasoro algebras}

\begin{definition} \cite{Va}
 A {\it PostLie algebra} is a Lie algebra $(\mathfrak{g},[,])$ with a bilinear product $\circ$ satisfying the following equations:
  \begin{equation}
    ((x\circ y)\circ z-x\circ(y\circ z))-((y\circ x)\circ z-y\circ(x\circ z))+[x,y]\circ z=0,\label{eq:polie1}
  \end{equation}
  \begin{equation}
    z\circ[x,y]-[z\circ x,y]-[x,z\circ y]=0,\label{eq:polie2}
  \end{equation}
  for all $x,y,z\in \mathfrak{g}$. We denote it by $(\frak g, [,],\circ)$.
\end{definition}

  \begin{lemma} \cite{BGN1}\label{lem:post}
    Let $(\mathfrak{g},[,])$ be a Lie algebra and $R: \mathfrak{g}\rightarrow \mathfrak{g}$ a Rota-Baxter operator of weight 1.
    Define a new operation $x\circ y = [R(x),y]$ on $\mathfrak{g}$. Then $(\mathfrak{g},[,],\circ)$ is a PostLie algebra.
  \end{lemma}

 Let $(\mathfrak{g},[,], \circ)$ be a PostLie algebra. Then the following operation (cf. \cite{BGN1})
  \begin{equation}\label{eq:PLie}
    \{x,y\}=x\circ y-y\circ x+[x,y],\;\;\forall x,y\in \frak g
  \end{equation}
  defines a Lie algebra structure on $\mathfrak{g}$.

\subsection{Induced PostLie algebras on the Witt algebra}

Direct application of Lemma~\ref{lem:post} gives the following.
  \begin{theorem}\label{Thm6.3}
The homogeneous Rota-Baxter operators of weight 1 with degree $0$ on the Witt algebra $W$ provided in Theorem~\ref{Thm2.21}
give rise to the following PostLie algebras $(W, [,],\circ)$, where $(W,[,])$ is the Witt algebra.
   \begin{enumerate}
      \item
      $
        L_m\circ^{\leqslant1} L_n=
        \begin{cases}
          -(m-n)L_{m+n} & m\geqslant2;\\
          0 & m\leqslant1.
        \end{cases}
      $
      \item
      $
        L_m\circ^{0} L_n=0
      ,\;\;\forall m,n\in\mathbb{Z}$.
      \item
      $
        L_m\circ^{>1} L_n=
        \begin{cases}
          -(m-n)L_{m+n} & m\leqslant1;\\
          0 & m\geqslant2.
        \end{cases}
      $
      \item
      $
        L_m\circ^{-{\rm Id}_W}L_n=-(m-n)L_{m+n}
      ,\;\;\forall m,n\in\mathbb{Z}$.
      \item
      $
        L_m\circ^{+,\alpha} L_n=
        \begin{cases}
          -(m-n)L_{m+n} & m<0; \\
          -\alpha n L_n & m=0; \\
          0 & m>0,
        \end{cases}
      $\\ where $\alpha\in\mathbb{C}$.
    \end{enumerate}
Moreover, these PostLie algebras are not mutually isomorphic.
 \end{theorem}

\begin{proof}
The conclusion follows from Lemma~\ref{lem:post} by direct
computation.
Note that
the Rota-Baxter operators
$R_0^{\leqslant1}$ and $R_0^{\geqslant-1}$
give the PostLie algebra (1);
the Rota-Baxter operator
$R_0^{0}$, i.e. $0$,
gives the PostLie algebra (2);
the Rota-Baxter operators
$R_0^{>1}$ and $R_0^{<-1}$
give the PostLie algebra (3);
the Rota-Baxter operator
$R_0^{\varnothing}$, i.e. $-{\rm Id}_W$,
gives the PostLie algebra (4);
the Rota-Baxter operators
$R_0^{+,\alpha}$ and $R_0^{-,\alpha}$
give the PostLie algebras (5).
In fact, the
PostLie algebras obtained by $R_0^{\geqslant-1}$, $R_0^{<-1}$ and $R_0^{-,\alpha}$
are isomorphic to the PostLie algebras obtained by
$R_0^{\leqslant1}$, $R_0^{>1}$ and $R_0^{+,\alpha}$ respectively through the linear transformation
$L_m\rightarrow-L_{-m}$. Moreover, it is also straightforward
to show that these PostLie algebras are not mutually isomorphic.
\end{proof}

  \begin{proposition}
    The PostLie algebras in Theorem~\ref{Thm6.3} give rise to the following Lie algebras under the bracket $\{,\}$ defined in Eq.~(\ref{eq:PLie}):
    \begin{enumerate}
      \item
      $
        \{L_m, L_n\}_{R_0^{\leqslant1}}=
        \begin{cases}
          -2(m-n)L_{m+n} & m,n\geqslant2; \\
          -(m-n)L_{m+n} & m\geqslant2, n\leqslant1;\\
          0 & m,n\leqslant1.
        \end{cases}
      $
      \item
      $
        \{L_m, L_n\}_{0}=0
      ,\;\;\forall m,n\in \mathbb{Z}$.
      \item
      $
        \{L_m, L_n\}_{R_0^{>1}}=
        \begin{cases}
          -2(m-n)L_{m+n} & m,n\leqslant1; \\
          -(m-n)L_{m+n} & m\leqslant1,n\geqslant2;\\
          0 & m,n\geqslant2.
        \end{cases}
      $
      \item
      $
        \{L_m, L_n\}_{-{\rm Id}_W}=-2(m-n)L_{m+n}
      ,\;\;\forall m,n\in \mathbb{Z}$.
      \item
      $
        \{L_m, L_n\}_{R_0^{+,\alpha}}=
        \begin{cases}
          -2(m-n)L_{m+n} & m,n<0; \\
          -(m-n)L_{m+n} & m<0,n>0; \\
          -(1-\alpha)mL_m & m<0,n=0;\\
          -\alpha n L_n & m=0,n>0; \\
          0 & m>0, n>0,
        \end{cases}
      $\\
      where $\alpha\in\mathbb{C}$.
    \end{enumerate}
  \end{proposition}

\begin{remark}
It is straightforward to verify the following.

The Lie algebra structures of $(W,\{,\}_{R_0^{\leqslant1}})$ and $(W,[,])$ are isomorphic on the subspace $\bigoplus_{m>1}W_m$.

The Lie algebra structures of $(W,\{,\}_{R_0^{>1}})$ and $(W,[,])$ are isomorphic on the subspace $\bigoplus_{m\leqslant1}W_m$.

The Lie algebra structures of $(W,\{,\}_{R_0^{+,\alpha}})$ and $(W,[,])$ are isomorphic on the subspace $\bigoplus_{m<0}W_m$.
\end{remark}

\subsection{Induced PostLie algebras on the Virasoro algebra}

For the case of the Virasoro algebra, the construction of Lemma~\ref{lem:post} provides the following result.
 \begin{theorem}\label{Thm6.6}
The homogeneous Rota-Baxter operators of weight 1 with degree $0$ on the Virasoro algebra $V$ given in Theorem~\ref{Thm3.17}
give rise to the following PostLie algebras $(V, [,],\circ)$, where $(V,[,])$ is the Virasoro algebra:
    \begin{enumerate}
      \item
      $
        L_m\circ^{\leqslant1,\mu} L_n
        \begin{cases}
          -(m-n)L_{m+n}-\tfrac{m^3-m}{12}\delta_{m+n,0}C & m\geqslant2;\\
          0 & m\leqslant1,
        \end{cases}  \\
        C\circ^{\leqslant1,\mu} L_n  = - \mu nL_n,\;\;
        L_m\circ^{\leqslant1,\mu} C  = 0,\;\;\forall m,n\in \mathbb{Z}
      $, where $\mu\in\mathbb{C}$.

      \item
      $
        L_m\circ^0 L_n  =
        C\circ^0 L_n  =
        L_m\circ^0 C  =0
      ,\;\;\forall m,n\in\mathbb{Z}$.

      \item
      $
        L_m\circ^{>1,\mu} L_n  =
        \begin{cases}
          -(m-n)L_{m+n}-\tfrac{m^3-m}{12}\delta_{m+n,0}C  & m\leqslant1;\\
          0 & m\geqslant2,
        \end{cases} \\
        C\circ^{>1,\mu} L_n  = - \mu nL_n,\;\;
        L_m\circ^{>1,\mu} C  = 0,\;\;\forall m,n\in \mathbb{Z}
        $, where $\mu\in\mathbb{C}$.

      \item
      $
        L_m\circ^{-{\rm Id}_V} L_n  =-(m-n)L_{m+n}-\tfrac{m^3-m}{12}\delta_{m+n,0}C$,\\
      $ C\circ^{-{\rm Id}_V}  L_n  =
        L_m\circ^{-{\rm Id}_V}  C  =0
      ,\;\;\forall m,n\in \mathbb{Z}$.

      \item
      $
        L_m\circ^{+,\alpha,\mu} L_n  =
        \begin{cases}
          -(m-n)L_{m+n}-\tfrac{m^3-m}{12}\delta_{m+n,0}C  & m<0; \\
          -\alpha n L_n & m=0; \\
          0 & m>0,
        \end{cases} \\
        C\circ^{+,\alpha,\mu}  L_n  = - \mu nL_n,\;\;
        L_m\circ^{+,\alpha,\mu}  C  = 0,\;\;\forall m,n\in \mathbb{Z}
        $, where $\alpha,\mu\in\mathbb{C}$.
    \end{enumerate}
  \end{theorem}
\begin{proof}
Note that since the images of $E_{c,0}$ and $E_{c,c}$ lie in the center of the Virasoro algebra $V$,
the PostLie algebras obtained by a homogeneous Rota-Baxter operator $R_0^{\ast,\theta,\mu,\nu}$ of weight 1 with degree $0$ on $V$ is independent from $\theta$ and $\nu$. Thus we denote by $(V,[,],\circ^{\ast,\mu})$ the PostLie algebra given by $R_0^{\ast,\theta,\mu,\nu}$ with fixed $\ast$ and $\mu$ .

Then the conclusion follows from Lemma~\ref{lem:post} by direct
computation.
Note that the Rota-Baxter operators
$R_0^{\leqslant1,0,\mu,\nu}, R_0^{\geqslant-1,0,\mu,\nu}$
give the PostLie algebras (1);
the Rota-Baxter operator
$R_0^{0,0,0,0}$, i.e. $0$,
gives the PostLie algebra (2);
the Rota-Baxter operators
$R_0^{>1,0,\mu,\nu}$, $R_0^{<-1,0,\mu,\nu}$
give the PostLie algebras (3);
the Rota-Baxter operator
$R_0^{\varnothing,0,0,-1}$, i.e. $-{\rm Id}_V$,
gives the PostLie algebra in (4);
the Rota-Baxter operators
$R_0^{+,\alpha,\theta,\mu,\nu}, R_0^{-,\alpha,\theta,\mu,\nu}$
give the PostLie algebras (5).
In fact, the PostLie algebras obtained by Rota-Baxter operators
$R_0^{\geqslant-1,0,\mu,\nu}, R_0^{<-1,0,\mu,\nu}, R_0^{-,\alpha,\theta,\mu,\nu}$
are isomorphic to the PostLie algebras obtained by Rota-Baxter operators
$R_0^{\leqslant1,0,\mu,\nu}, R_0^{>1,0,\mu,\nu}, R_0^{+,\alpha,\theta,\mu,\nu}$
respectively through the linear transformation of basis
$L_m\rightarrow-L_{-m}, C\rightarrow-C$.
Moreover, it is also straightforward to show that these PostLie algebras are not mutually isomorphic.
\end{proof}

  \begin{proposition}
    The PostLie algebras in Theorem~\ref{Thm6.6} give rise to the following Lie algebras under the bracket  $\{,\}$ defined in Eq.~(\ref{eq:PLie}):
    \begin{enumerate}
      \item
      $
        \{L_m, L_n\}_{\leqslant1,\mu}
        \begin{cases}
          -2(m-n)L_{m+n}-\tfrac{m^3-m}{6}\delta_{m+n,0}C & m,n\geqslant2;\\
          -(m-n)L_{m+n}-\tfrac{m^3-m}{12}\delta_{m+n,0}C & m\geqslant2,n\leqslant1;\\
          0 & m,n\leqslant1,
        \end{cases}  \\
        \{C, L_n\}_{\leqslant1,\mu}  = - \mu nL_n
      ,\;\;\forall m,n\in\mathbb{Z}$, where $\mu\in\mathbb{C}$.

      \item
      $
        \{L_m, L_n\}_0  =
        \{C, L_n\}_0  = 0
      ,\;\;\forall m,n\in\mathbb{Z}$.

      \item
      $
        \{L_m, L_n\}_{>1,\mu}  =
        \begin{cases}
          -2(m-n)L_{m+n}-\tfrac{m^3-m}{6}\delta_{m+n,0}C  & m,n\leqslant1;\\
          -(m-n)L_{m+n}-\tfrac{m^3-m}{12}\delta_{m+n,0}C  & m\leqslant1,n\geqslant2;\\
          0 & m,n\geqslant2,
        \end{cases} \\
        \{C, L_n\}_{>1,\mu}  = - \mu nL_n
      ,\;\;\forall m,n\in\mathbb{Z}$, where $\mu\in\mathbb{C}$.

      \item
      $
        \{L_m, L_n\}_{-{\rm Id}_V}  = -2(m-n)L_{m+n}-\tfrac{m^3-m}{6}\delta_{m+n,0}C,\;\;
        \{C, L_n\}_{-{\rm Id}_V} = 0
      ,\;\;\forall m,n\in\mathbb{Z}$.

      \item
      $
        \{L_m, L_n\}_{+,\alpha,\mu}  =
        \begin{cases}
          -2(m-n)L_{m+n}-\tfrac{m^3-m}{6}\delta_{m+n,0}C  & m,n<0; \\
          -(m-n)L_{m+n}-\tfrac{m^3-m}{12}\delta_{m+n,0}C  & m<0,n>0; \\
          -(1-\alpha)mL_m-\tfrac{m^3-m}{12}\delta_{m,0}C  & m<0,n=0; \\
          -\alpha n L_n & m=0,n>0; \\
          0 & m,n>0,
        \end{cases} \\
        \{C, L_n\}_{+,\alpha,\mu}  = - \mu nL_n
      ,\;\;\forall m,n\in\mathbb{Z}$
      , where $\alpha,\mu\in\mathbb{C}$.
    \end{enumerate}
  \end{proposition}
\begin{remark}
It is straightforward to verify the following.

The Lie algebra structures of $(V,\{,\}_{\leqslant1,0})$ and $(V,[,])$ are isomorphic on the subspace $\bigoplus_{m>1}V_m$.

The Lie algebra structures of $(V,\{,\}_{>1,0})$ and $(V,[,])$ are isomorphic on the subspace $\bigoplus_{m\leqslant1}V_m$.

The Lie algebra structures of $(V,\{,\}_{-{\rm Id}_V})$ and $(V,[,])$ are isomorphic.

The Lie algebra structures of $(V,\{,\}_{+,\alpha,0})$ and $(V,[,])$ are isomorphic on the subspace $\bigoplus_{m<0}V_m$.
\end{remark}

\section*{Acknowledgments} C. Bai thanks the support by NSFC
(grant nos. 11271202, 11221091, 11425104) and SRFDP (grant no. 20120031110022). N. Jing
thanks the partial support of Simons Foundation (grant no. 198129) and
NSFC (grant nos. 11271138, 11531004).

\bibliographystyle{amsplain}

\begin{thebibliography}{99}

\bibitem{Bai1} C. Bai, {\em A unified algebraic approach to the classical Yang--Baxter equation}, J. Phys. A: Math. Gen.  40 (2007) 11073-11082.

\bibitem{BGN1} C. Bai, L. Guo, and X. Ni, {\em Nonabelian generalized Lax pairs, the classical Yang-Baxter
equation and PostLie algebras}, Comm. Math. Phys. 297 (2010) 553-596.

\bibitem{BGN2} C. Bai, L. Guo, and X. Ni, {\em Generalizations of the classical Yang-Baxter equation
and $\mathcal{O}$-operators}, J. Math. Phys. 52 (2011) 063515, 17 pp.

\bibitem{Baxter} G. Baxter, {\em An analytic problem whose solution follows from a simple algebraic identity}, Pacific J. Math. 10 (1960) 731-742.

\bibitem{Belavin} A.~A. Belavin, {\em Dynamical symmetry of integrable quantum systems}, Nuclear Phys. B 180 (1981) 189-200.

\bibitem{Bu} D. Burde, {\em Left-symmetric algebras, or pre-Lie algebras in geometry and physics}, Cent. Eur. J. Math. 4 (2006) 323-357.

\bibitem{CP} V. Chari and A. Pressley, {\em A guide to quantum groups}, Cambridge University Press, Cambridge, 1994.

\bibitem{G} M. Gerstenhaber, {\em The cohomology structure of an associative ring}, Ann. Math. 78 (1963) 267-288.

\bibitem{GS} I.~Z.  Golubschik and V.~V. Sokolov, {\em Generalized operator Yang-Baxter
equations, integrable ODES and nonassociative algebras}, J. Nonlin.
Math. Phys. 7 (2000) 184-197.

\bibitem{Guo1} L. Guo, {\em What is a Rota-Baxter algebra?} Notices Amer. Math. Soc. 56 (2009) 1436-1437.

\bibitem{Guo2} L. Guo, {\em An introduction to Rota-Baxter algebra}, International Press, Somerville, MA and Higher Education Press, Beijing, 2012.

\bibitem{Ko} J.-L. Koszul, {\em Domaines born\'es homog\`enes et orbites de groupes de transformation affines}, Bull. Soc. Math. France 89 (1961) 515-533.

\bibitem{PBG} J. Pei, C. Bai, and L. Guo, {\em Rota-Baxter operators on sl(2,C) and solutions of the classical Yang-Baxter equation}, J. Math. Phys. 55 (2014) 021701, 17 pp.

\bibitem{Rota-Kung} G.-C. Rota, {\em Baxter operators, an introduction},
Gian-Carlo Rota on combinatorics, pp. 504-512,
Contemp. Mathematicians, Birkh\"auser Boston, Boston, 1995.

\bibitem{S} M.~A. Semenov-Tian-Shansky, {\em What is a classical r-matrix?} Funct. Anal. Appl. 17 (1983) 259-272.

\bibitem{Va} B. Vallette, {\em Homology of generalized partition posets}, J. Pure Appl. Algebra 208 (2007) 699-725.


\bibitem{Vi} E.~B. Vinberg, {\em The theory of convex homogeneous cones}, Trans. Moscow Math. Soc. (1963) 340-403.

\end{thebibliography}

\end{document}